\documentclass[12pt, journal, letterpaper, onecolumn, draftclsnofoot]{IEEEtran}

\IEEEoverridecommandlockouts
\usepackage{amsmath,amssymb,amsthm}
\usepackage{graphicx}
\usepackage{subfigure}
\usepackage{curves}
\usepackage{amsfonts}
\usepackage{epsfig}
\usepackage{stfloats}
\usepackage{cite}
\usepackage{algorithm}
\usepackage{algorithmic}
\usepackage{epstopdf}

\begin{document}

\title{Secure Polar Coding for the Two-Way \\Wiretap Channel}
\author{Mengfan~Zheng, Meixia~Tao, Wen~Chen and Cong~Ling
\thanks{M. Zheng, M. Tao and W. Chen are with the Department of
Electronic Engineering at Shanghai Jiao Tong University, Shanghai, China. Emails: \{zhengmengfan, mxtao, wenchen\}@sjtu.edu.cn. C. Ling is with the Department of Electrical and Electronic Engineering at Imperial College London, United Kingdom. Email: c.ling@imperial.ac.uk.

}
}

\maketitle

\begin{abstract}
 We consider the problem of polar coding for secure communications over the two-way wiretap channel, where two legitimate users communicate with each other simultaneously while a passive eavesdropper overhears a combination of their exchanged signals. The legitimate users wish to design a cooperative jamming code such that the interference between their codewords can jam the eavesdropper. In this paper, we design a polar coded cooperative jamming scheme that achieves the whole secrecy rate region of the general two-way wiretap channel under the strong secrecy criterion. The chaining method is used to make proper alignment of polar indices. The randomness required to be shared between two legitimate users is treated as a limited resource and we show that its rate can be made negligible by increasing the blocklength and the number of chained blocks. For the special case when the eavesdropper channel is degraded with respect to the legitimate ones, a simplified scheme is proposed which can simultaneously ensure reliability and weak secrecy within a single transmission block. An example of the binary erasure channel case is given to demonstrate the performance of our scheme.
\end{abstract}


\section{Introduction}
\label{S:Intro}
  Wyner proved in \cite{wyner1975wire} that it is possible to communicate both reliably and securely over a wiretap channel, on the premise that the eavesdropper channel is degraded with respect to the legitimate channel. Since then, numerous works have been done on showing the existence of secure coding schemes for different kinds of channels. However, few of these results provide guidance for designing a specific polynomial-time coding scheme, except for some special cases \cite{suresh2010strong,thangaraj2007applications,cheraghchi2012invertible}.
  Polar codes, proposed by Ar{\i}kan \cite{arikan2009channel}, have demonstrated capacity-achieving property in both source and channel coding \cite{arikan2009channel,arikan2010source,korada2009polar,korada2010lossy,honda2013asymmetric}. The principle that lies behind polar codes is that one can generate a series of extremal channels (noiseless or purely noisy) from repeated uses of a single-user channel. The structure of polar codes makes them also suitable for designing secrecy codes. Polar coding has been studied for wiretap channels  \cite{hof2010secrecy,andersson2010nested,mahdavifar2011achieving,koyluoglu2012polar,sasoglu2013strong,wei2016general,chou2016broad,gulcu2017wiretap}, fading wiretap channels \cite{Si2016fadingwiretap}, multiple access wiretap channels \cite{chou2016macwiretap,hajimomeni2016mawc}, and broadcast channels with confidential messages \cite{chou2016broad,gulcu2017wiretap}. It is shown that polar codes can achieve the secrecy capacity of the general wiretap channel under the strong secrecy criterion \cite{chou2016broad,gulcu2017wiretap}.

  The two-way wiretap channel models the situation when two legitimate users communicate with each other simultaneously in the presence of a passive eavesdropper. In this model, signals overheard by the eavesdropper are combinations of the exchanged signals between two legitimate users. This motivates the idea of leveraging interference between two users' transmitted codewords to degrade the eavesdropper channel, known as \textit{coded cooperative jamming}. This problem was first investigated in \cite{tekin2009mactw}, and the achievable rate region for the two-way wiretap channel was derived in \cite{pierrot2011twoway,Gamal2013secrecy}. A practical scheme based on low-density parity-check (LDPC) code was presented in \cite{pierrot2012ldpc}, which can guarantee weak secrecy for the special case of binary-input Gaussian two-way wiretap channel with equal-gained interference.
  
  Note that the eavesdropper sees a 2-user multiple access channel (MAC) in the two-way wiretap channel. Polar coding for MACs has been studied in \cite{sasoglu2013mac,abbe2012mmac,Nasser2016MAC,arikan2012sw,onayscmac,mahdavifar2016uniform}. There are two types of MAC polarization methods in literature, either synthesizing $N$ uses of the original MAC into $N$ new extremal MACs \cite{sasoglu2013mac,abbe2012mmac,Nasser2016MAC}, or $2N$ extremal point-to-point channels \cite{arikan2012sw,onayscmac,mahdavifar2016uniform}. In our scheme, we adopt Ar{\i}kan's monotone chain rule expansion method \cite{arikan2012sw} which belongs to the first type, as it can achieve all points on the dominant face of the achievable rate region of a MAC without time sharing, and has simple structure and low encoding/decoding complexity. 
  
  In this paper, we use polar codes to design a coded cooperative jamming scheme for the general two-way wiretap channel. The chaining method \cite{hassani2014universal} is used to deal with unaligned polar indices for the general channel cases. The main contributions of this paper include:
  \begin{itemize}
  	\item A polar coded cooperative jamming scheme for the general two-way wiretap channel is proposed, without any constraint on channel symmetry or degradation. Self interference of each user is considered in the code design, making our proposed scheme suitable for a large variety of channels rather than additive ones. For additive channels, one may assume that each user's self interference can be perfectly canceled. However, under a general setting, this assumption is inappropriate. In this paper, we treat self interference as side information of legitimate channels, which is involved in the polar code design, encoding and decoding.
  	\item Information theoretical analysis on reliability, secrecy and achievable rate region is performed. Instead of assuming channel prefixing can be done perfectly, we apply polar coding to do channel prefixing and show that the induced joint distribution of random variables involved in the coding scheme is asymptotically indistinguishable from the target one. The amount of randomness required to be shared between two legitimate users is considered as a limited resource, and we show that its rate can be made arbitrarily small by increasing the blocklength and chaining sufficient number of blocks in our scheme. By applying MAC polarization on the eavesdropper channel using different types of monotone chain rule expansions, we can achieve different secrecy rate pairs. We prove that our proposed scheme can achieve all points on the dominant face of the secrecy rate region of a two-way wiretap channel under the strong secrecy criterion.
  	\item A single-block scheme for the special case of degraded two-way wiretap channel is provided. In the case when the eavesdropper channel is degraded with respect to both legitimate channels, we show that with a slight modification, our proposed scheme can achieve the secrecy rate region under the weak secrecy criterion within a single transmission block.
  	\item An example of the binary erasure channel case is presented to evaluate the performance of our proposed scheme for different code lengths. The information leakage, block error rate and secrecy sum rate are estimated for code length $2^8$ to $2^{27}$. The results confirm the secrecy rate-achieving capability of our proposed scheme.
  \end{itemize}

  The rest of this paper is organized as follows. In Section \ref{S:II} we introduce the two-way wiretap channel model and state the problem we investigate. Section \ref{S:III} provides some necessary background on polarization and polar codes. In Section \ref{S:IV} we describe details of our proposed polar coding scheme and analyze its performance. Section \ref{S:V} shows a special case when weak secrecy can be obtained within a single transmission block. Section \ref{S:VI} gives an example of the binary erasure channel case. We conclude this paper in Section \ref{S:VII}.
  
  \textit{Notation:} $[N]$ denotes the index set of $\{1,2,...,N\}$. Random variables are denoted by capital letters $X$, $Y$, $U$, $V$, ... with values $x$, $y$, $u$, $v$, ... respectively. For a vector $\mathbf{y}=(y^1,y^2,...y^N)$, $y^{i:j}$ denotes its subvector $(y^i,...,y^j)$, and $y^{\mathcal{A}}$ ($\mathcal{A}\subset [N]$) denotes its subvector $\{y^i:i\in\mathcal{A}\}$. $\mathbf{F}^{\otimes n}$ denotes the $n^{th}$ Kronecker power of $\mathbf{F}$. $\mathbf{G}_N=\mathbf{B}_N \textbf{F}^{\otimes n}$ is the generator matrix of polar codes \cite{arikan2009channel}, where $N=2^n$ is the code length with $n$ being an arbitrary integer, $\mathbf{B}_N$ is a permutation matrix known as bit-reversal matrix, and $\textbf{F}=
  \begin{bmatrix}
  1 & 0 \\
  1 & 1
  \end{bmatrix}$.

\section{Problem Statement}
\label{S:II}
 \subsection{Channel Model}

  We consider the secure communication problem in the two-way wiretap channel as illustrated in Fig. \ref{fig:2way}. In this model, each of the two legitimate users, Alice and Bob, is equipped with a transmitter and a receiver. The channel is assumed to be full-duplex, and the two users communicate with each other simultaneously under the existence of a passive eavesdropper, Eve. Details of the communications are as follows:
  \begin{itemize}
  	\item Alice wants to send a message $M_1$ to Bob at rate $R_1$ over $N$ channel uses, she encodes $M_1$ into a codeword $\mathbf{X}_1$ and transmits it through the channel;
  	\item Bob wants to send a message $M_2$ to Alice at rate $R_2$ also over $N$ channel uses, he encodes $M_2$ into a codeword $\mathbf{X}_2$ and transmits it through the channel;
  	\item Alice observes $\mathbf{Y}_2$ from the channel and recovers $\hat{M}_2$;
  	\item Bob observes $\mathbf{Y}_1$ from the channel and recovers $\hat{M}_1$;
  	\item Eve observes $\mathbf{Y}_e$.
  \end{itemize}
\begin{figure}[tb]
	\centering
	\includegraphics[width=11cm]{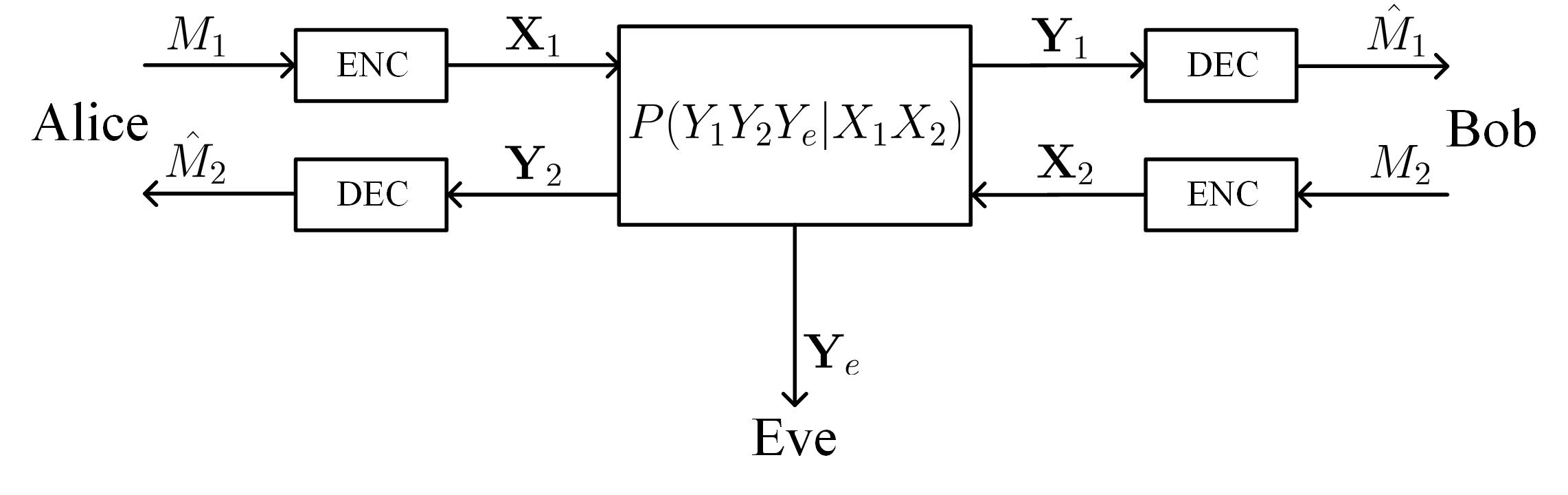}
	\caption{The two-way wiretap channel.} \label{fig:2way}
\end{figure}

  \newtheorem{definition}{Definition}
  \begin{definition}
  	 A memoryless two-way wiretap channel $\big{(}\mathcal{X}_1,\mathcal{X}_2,\mathcal{Y}_1,\mathcal{Y}_2,\mathcal{Y}_e,P_{Y_1Y_2Y_e|X_1X_2}\big{)}$ consists of two input alphabets $\mathcal{X}_1$ and $\mathcal{X}_2$, three output alphabets $\mathcal{Y}_1$, $\mathcal{Y}_2$ and $\mathcal{Y}_e$, and transition probability $P_{Y_1Y_2Y_e|X_1X_2}$ such that
  \begin{equation}
    \forall (x_1,x_2)\in \mathcal{X}_1\times \mathcal{X}_2,
    \sum_{y_1\in\mathcal{Y}_1}\sum_{y_2\in\mathcal{Y}_2}\sum_{y_e\in\mathcal{Y}_e}P_{Y_1Y_2Y_e|X_1X_2}(y_1,y_2,y_e|x_1,x_2)=1.
  \end{equation}
  \end{definition}

 \subsection{Coded Cooperative Jamming}

  The multiple access nature of the eavesdropper channel renders Alice and Bob an advantage over Eve, since the combination of their signals may have a detrimental effect on her. In the multi-user communication scenario, a natural approach to enhance security is to use \textit{cooperative jamming} \cite{tekin2009mactw}. While one user is transmitting secret messages, the other user transmits artificial noise to reduce the eavesdropper's signal-to-noise ratio.  In such a scheme, only one user can transmit useful information at a time. Another approach which overcomes this limitation is to utilize interference between codewords to jam the eavesdropper. In this case, both users transmit secret messages simultaneously, and their codewords are elaborately designed so that the interference between them can confuse the eavesdropper. This scheme is called coded cooperative jamming \cite{tekin2009mactw}. In this paper, we use polar codes to design such a code.

  The goal of designing a secure coding scheme for the two-way wiretap channel is to make sure Eve obtains no (or vanishing) information about $M_1$ and $M_2$ from $\mathbf{Y}_e$, while Alice and Bob can estimate their intended massages correctly. The performance of a coding scheme is assessed by its reliability and secrecy. For a coding scheme of blocklength $N$, reliability is measured by the probability of error
  \begin{equation}
  P_e(N)=\mathrm{Pr}\big{\{}(\hat{M}_1,\hat{M}_2) \neq (M_1,M_2)\big{\}}.
  \end{equation}
  Secrecy can be measured by the information leakage
  \begin{equation}
  L(N)=I(\mathbf{Y}_e;M_1,M_2),
  \end{equation}
  or the information leakage rate
  \begin{equation}
  L_R(N)=\frac{1}{N}L(N).
  \end{equation}
  The objective of a secure coding scheme is then:
  \begin{equation}
  \label{RC}
  \lim\limits_{N \to \infty}{P_e(N)}=0;
  \end{equation}
  \begin{equation}
  \label{SSC}
  \lim\limits_{N \to \infty}{L(N)}=0 \text{ (strong secrecy), or}
  \end{equation}
  \begin{equation}
  	\label{WSC}
  	\lim\limits_{N \to \infty}{L_R(N)}=0 \text{ (weak secrecy)}.
  \end{equation}

  Criterion (\ref{WSC}) is called weak secrecy because it does not guarantee vanishing information leakage. For some strict situations this is unacceptable. In Section \ref{S:IV} we will introduce the general strong secrecy scheme. In Section \ref{S:V} we discuss a special case when weak secrecy can be achieved within a single transmission block.

  \subsection{Achievable Rate Region}
  A rate pair $(R_1,R_2)$ is said to be achievable for a two-way wiretap channel under the strong/weak secrecy criterion if there exists a coding scheme such that (\ref{RC}) and (\ref{SSC})/(\ref{WSC}) can be satisfied. The achievable rate region of this channel is the closure of all achievable rate pairs. For a two-way wiretap channel with transition probability $P_{Y_1Y_2Y_e|X_1X_2}(y_1,y_2,y_e|x_1,x_2)$, the secrecy rate region under the strong (as well as weak) secrecy criterion is \cite{pierrot2011twoway}
  \begin{equation}
  \label{SSR}
  \mathcal{R}_S(P_{Y_1Y_2Y_e|X_1X_2}) = \bigcup_{P\in\mathcal{P}}
  \left\lbrace
  \begin{matrix}
  	\left(
  	\begin{array}{ccc}
  		R_1\\
  		R_2
  	\end{array}
  	\right)&
  	\left|
  	\begin{array}{ccc}
  		\begin{array}{ccc}
  			R_1 \leq I(Y_1;C_1|X_2)-I(C_1;Y_e)\\
  			R_2 \leq I(Y_2;C_2|X_1)-I(C_2;Y_e)\\
  			R_1+R_2 \leq I(Y_1;C_1|X_2)+I(Y_2;C_2|X_1)\\
  			~~~~~~~~~~~~~~~~~~-I(C_1,C_2;Y_e)
  		\end{array}
  	\end{array}\right.
  \end{matrix}
  \right\rbrace,
  \end{equation}
  where
  \begin{equation*}
  \mathcal{P}=\{P_{X_1X_2C_1C_2Y_1Y_2Y_e} \text{ factorizing as: } P_{Y_1Y_2Y_e|X_1X_2}P_{X_1|C_1}P_{C_1}P_{X_2|C_2}P_{C_2}\}.
  \end{equation*}

\section{Review of Polar Coding}
\label{S:III}

   \subsection{Polar Coding for Asymmetric Channels}
   \label{S:IIIA}
   In this subsection we review the polar coding scheme proposed in \cite{honda2013asymmetric} and simplified in \cite{chou2016broad,gad2016asymmetric} for asymmetric channels. Consider $N$ independent uses of a binary-input discrete memoryless channel (B-DMC) $P_{Y|X}(y|x)$, where $X$ is binary with arbitrary distribution and $Y$ is defined on an arbitrary countable alphabet. Let $U^{1:N}=X^{1:N}\mathbf{G}_N$. For $\delta_N=2^{-N^\beta}$ with $\beta \in (0,1/2)$, define the following polarized sets:
   \begin{align}
   \mathcal{H}^{(N)}_X&=\{i\in [N]:Z(U^i|U^{1:i-1})\geq 1-\delta_N\},\label{HX}\\
   \mathcal{L}^{(N)}_X&=\{i\in [N]:Z(U^i|U^{1:i-1})\leq \delta_N\},\label{LX}\\
   \mathcal{H}^{(N)}_{X|Y}&=\{i\in [N]:Z(U^i|Y^{1:N},U^{1:i-1})\geq 1-\delta_N\},\label{HXY}\\
   \mathcal{L}^{(N)}_{X|Y}&=\{i\in [N]:Z(U^i|Y^{1:N},U^{1:i-1})\leq \delta_N\},\label{LXY}
   \end{align}
   where $Z(X|Y)$ is the Bhattacharyya parameter of a random variable pair $(X,Y)$, defined as
   	\begin{equation}
   	Z(X|Y)=2\sum_{y\in \mathcal{Y}} P_Y(y)\sqrt{P_{X|Y}(0|y)P_{X|Y}(1|y)}.
   	\end{equation}
   It is shown that \cite{arikan2010source}
   \begin{equation}
   \begin{aligned}
   \lim_{N\rightarrow \infty}\frac{1}{N}|\mathcal{H}^{(N)}_X|&=H(X),~~~~~~
   \lim_{N\rightarrow \infty}\frac{1}{N}|\mathcal{L}^{(N)}_X|=1-H(X),\\
   \lim_{N\rightarrow \infty}\frac{1}{N}|\mathcal{H}^{(N)}_{X|Y}|&=H(X|Y),~~
   \lim_{N\rightarrow \infty}\frac{1}{N}|\mathcal{L}^{(N)}_{X|Y}|=1-H(X|Y).
   \end{aligned}
   \end{equation}

   To construct a polar code for $W$, partition indices of $U^{1:N}$ into the following sets:
   \begin{equation}
   \begin{aligned}
   \mathcal{I}&\triangleq \mathcal{H}_X^{(N)}\cap \mathcal{L}_{X|Y}^{(N)}, \\
   \mathcal{F}&\triangleq \mathcal{H}_X^{(N)}\cap (\mathcal{L}_{X|Y}^{(N)})^C, \\
   \mathcal{D}&\triangleq (\mathcal{H}_X^{(N)})^C.
   \end{aligned}
   \end{equation}
   Since $\{u^i\}_{i\in \mathcal{I}}$ are uniformly distributed and can be reliably decoded, they will be filled with uniformly distributed information bits. For $\{u^i\}_{i\in \mathcal{F}\cup \mathcal{D}}$, reference \cite{honda2013asymmetric} suggests to assign them by random mappings $\lambda_{\mathcal{I}^C}\triangleq \{\lambda_i\}_{i\in\mathcal{I}^C}$ that sample distribution $P_{U^i|U^{1:i-1}}$, which are shared between the encoder and the decoder. However, Exchanging the shared randomness may heavily increase the encoder's overhead since the non-information bits usually form a large portion of the uncoded bits. A simplified scheme which only requires a vanishing rate of shared randomness was independently proposed in \cite{chou2016broad} and \cite{gad2016asymmetric}, which is summarized as follows.
   \begin{itemize}
   	\item $\{u^i\}_{i\in \mathcal{I}}$ carry uniformly distributed information bits,
   	\item $\{u^i\}_{i\in \mathcal{F}}$ are filled with uniformly distributed frozen bits (shared between the encoder and the decoder),
   	\item $\{u^i\}_{i\in \mathcal{D}}$ are assigned by random mappings:
   	\begin{equation*}
   	\label{randommapping}
   	u^i=\begin{cases}
   	0 ~~\text{  w.p. } P_{U^i|U^{1:i-1}}(0|u^{1:i-1}),\\
   	1 ~~\text{  w.p. } P_{U^i|U^{1:i-1}}(1|u^{1:i-1}).
   	\end{cases}
   	\end{equation*}
   	\item Codeword $x^{1:N}=u^{1:N}\mathbf{G}_N$ is transmitted to the receiver.
   	\item $\{u^i\}_{i\in (\mathcal{H}_X^{(N)})^C\cap (\mathcal{L}_{X|Y}^{(N)})^C}$ is separately transmitted to the receiver with some reliable error-correcting code.
   \end{itemize}

   It is shown in \cite{chou2016broad,gad2016asymmetric} that the rate of the shared almost deterministic bits in $(\mathcal{H}_X^{(N)})^C\cap (\mathcal{L}_{X|Y}^{(N)})^C$ vanishes as $N$ goes large. Having received $y^{1:N}$ and recovered the shared bits, the receiver decodes $u^{1:N}$ with a successive cancellation (SC) decoder:
   \begin{equation*}
   \bar{u}^{i}=
   \begin{cases}
   u^i,&\text{if } i\in (\mathcal{L}_{X|Y}^{(N)})^C\\
   \arg\max_{u\in\{0,1\}}P_{U^{i}|Y^{1:N}U^{1:i-1}}(u|y^{1:N},u^{1:i-1}),&\text{if } i\in \mathcal{L}_{X|Y}^{(N)}
   \end{cases}.
   \end{equation*}
   
   The rate this scheme, $R=\frac{1}{N}|\mathcal{I}|$, satisfies
   \begin{equation}
   \lim_{N\rightarrow \infty}R=I(X;Y),
   \end{equation}
   and the block error probability can be upper bounded by
   \begin{equation}
   \label{SC-EP}
   P_e \leq\sum_{i\in\mathcal{L}_{X|Y}^{(N)}}{Z(U^i|Y^{1:N},U^{1:i-1})}=O(N2^{-N^\beta}).
   \end{equation}

   \subsection{Polar Coding for Multiple Access Channels}
   \label{Sec-MAC}
   In this subsection we recap the monotone chain rule expansion based MAC polarization method introduced in \cite{arikan2012sw} and generalized to asymmetric channels in \cite{zheng2016polarIC}.  The achievable rate region of a binary-input discrete memoryless 2-user MAC $P_{Y|X_1X_2}(y|x_1,x_2)$  is given by \cite{cover2012informtaion}
   \begin{equation}
   \mathcal{R}(P_{Y|X_1X_2})=\left\lbrace
   \begin{matrix}
   \left(
   \begin{array}{ccc}
   R_1\\
   R_2
   \end{array}
   \right)&
   \left|
   \begin{array}{ccc}
   \begin{array}{ccc}
   0 \leq R_1 \leq I(X_1;Y|X_2)\\
   0 \leq R_2 \leq I(X_2;Y|X_1)\\
   R_1+R_2 \leq I(X_1,X_2;Y)
   \end{array}
   \end{array}\right.
   \end{matrix}
   \right\rbrace .
   \end{equation}
   
   Define
   \begin{equation}
   \label{UX}
   U_1^{1:N}=X_1^{1:N}\mathbf{G}_N,~~U_2^{1:N}=X_2^{1:N}\mathbf{G}_N,
   \end{equation}
   and let $S^{1:2N}=(S^1,...,S^{2N})$ be a permutation of $U_1^{1:N}U_2^{1:N}$ such that it preserves the relative order of elements of both $U_1^{1:N}$ and $U_2^{1:N}$, called a \textit{monotone chain rule expansion}. For $i\in [2N]$, let $b_i=0$ represent that $S^i\in U_1^{1:N}$, and $b_i=1$ represent that $S^i\in U_2^{1:N}$. Then a monotone chain rule expansion can be represented by a string $\mathbf{b}_{2N}=b_1b_2...b_{2N}$, called the \textit{path} of the expansion. The mutual information between the receiver and two users can be expanded as
   \begin{align*}
   \label{MCRE}
   I(Y^{1:N};U_1^{1:N},U_2^{1:N})&=H(U_1^{1:N},U_2^{1:N})-H(U_1^{1:N},U_2^{1:N}|Y^{1:N})\\
   &=NH(X_1)+NH(X_2)-\sum_{i=1}^{2N}H(S^i|Y^{1:N},S^{1:i-1}),
   \end{align*}
   and the rate of user $j$ ($j=1,2$) is
   \begin{equation}
   R_{U_j}=H(X_j)-\frac{1}{N}\sum_{i\in \mathcal{S}_{U_j}}H(S^i|Y^{1:N},S^{1:i-1}),
   \end{equation}
   where $\mathcal{S}_{U_j}\triangleq \{ i\in [2N]:S^i\in U_j^{1:N}\}$. $(R_{U_1},R_{U_2})$ is a point on the dominant face of $\mathcal{R}(P_{Y|X_1X_2})$. It is shown that  arbitrary points on the dominant face can be achieved with expansions of type $0^i1^N0^{N-i}$ ($0\leq i\leq N$) given sufficiently large $N$ \cite{arikan2012sw}. It is also shown that $H(S^i|Y^{1:N},S^{1:i-1})$ ($i\in [2N]$) polarizes to 0 or 1 as $N$ goes to infinity.
   
   Having selected a specific expansion for a target rate pair, we still need enough code length to polarize the MAC sufficiently. In order to do so, we need to scale the path. For any integer $l=2^m$, let $l\mathbf{b}_{2N}$ denote $$\underbrace{b_1\cdots b_1}_l \underbrace{b_2\cdots b_2}_l\cdots\cdots \underbrace{b_{2N}\cdots b_{2N}}_l,$$ which is a monotone chain rule for $U_1^{1:lN}U_2^{1:lN}$. It is shown in \cite{arikan2012sw} that $\mathbf{b}_{2N}$ and $l\mathbf{b}_{2N}$ have the same rate pair. With this result, we can construct a polar code for the 2-user MAC using point-to-point polar codes. For $j=1,2$, let $f_j(i):[N]\rightarrow \mathcal{S}_{U_j}$ be the mapping from indices of $U_j^{1:N}$ to those of $S^{\mathcal{S}_{U_j}}$. Define
   \begin{align}
   \mathcal{H}^{(N)}_{S_{U_j}}&\triangleq \{i\in [N]:Z(S^{f_j(i)}|S^{1:f_j(i)-1})\geq 1-\delta_N\},\\
   \mathcal{L}^{(N)}_{S_{U_j}|Y}&\triangleq \{i\in [N]:Z(S^{f_j(i)}|Y^{1:N},S^{1:f_j(i)-1})\leq \delta_N\},
   \end{align}
   for $\delta_N=2^{-N^\beta}$ with $\beta \in (0,1/2)$. Since $X_1$ and $X_2$ are independent, we have
   \begin{equation}
   \mathcal{H}^{(N)}_{S_{U_j}}=\mathcal{H}^{(N)}_{X_j}\triangleq \{i\in [N] :Z(U_j^i|U_j^{1:i-1})\geq 1-\delta_N\}.
   \end{equation}
   Then we can partition indices of $U_j^{1:N}$ into
   \begin{equation}
   \begin{aligned}
   \mathcal{I}_j&\triangleq \mathcal{H}^{(N)}_{S_{U_j}}\cap \mathcal{L}^{(N)}_{{S_{U_j}}|Y}, \\
   \mathcal{F}_j&\triangleq \mathcal{H}^{(N)}_{S_{U_j}}\cap (\mathcal{L}^{(N)}_{{S_{U_j}}|Y})^C, \\
   \mathcal{D}_j&\triangleq (\mathcal{H}^{(N)}_{S_{U_j}})^C,
   \end{aligned}
   \end{equation}
   and then apply the polar coding scheme introduced in the previous subsection. The receiver jointly decodes $U_1^{1:N}$ and $U_2^{1:N}$ with a successive cancellation decoder according to the permutation used.
   
   \newtheorem{proposition}{Proposition}
   \begin{proposition}[\cite{arikan2012sw}]
   	\label{proposition:MAC}
   	Let $P_{Y|X_1X_2}(y|x_1,x_2)$ be the transition probability of a binary-input memoryless 2-user MAC. Consider the transformation defined in (\ref{UX}). Let $N_0=2^{n_0}$ for some $n_0\geq 1$ and fix a path $\mathbf{b}_{2N_0}$ for $U_1^{1:N_0}U_2^{1:N_0}$. The rate pair for $\mathbf{b}_{2N_0}$ is denoted by $(R_{U_1},R_{U_2})$. Let $N=2^lN_0$ for $l\geq 1$ and let $S^{1:2N}$ be the expansion represented by $2^l\mathbf{b}_{2N_0}$. Then, for any given $\delta>0$, as $l$ goes to infinity, we have
   	\begin{equation}
   	\begin{aligned}
   	&\frac{1}{2N}\big{|}\{ i\in [2N]:\delta <Z(S^i|Y^{1:N},S^{1:i-1})<1-\delta \}\big{|}\rightarrow 0,\\
   	&~~~~~~~~\frac{|\mathcal{I}_1|}{N}\rightarrow R_{U_1} \text{ and } \frac{|\mathcal{I}_2|}{N}\rightarrow R_{U_2}. \label{MAC:rate}
   	\end{aligned}
   	\end{equation}
   \end{proposition}

\section{Polar Coding for the Two-Way Wiretap Channel}
\label{S:IV}
\subsection{The Proposed Scheme}
\label{S:IV-A}
\subsubsection{Polarization of Legitimate Channels}
\label{SectionIV1}
For a given $P_{X_1X_2C_1C_2Y_1Y_2Y_e}\in \mathcal{P}$, define Bob's effective channel as
\begin{equation*}
W_1(y_1|c_1,x_2)\triangleq \sum_{x_1}P_{Y_1|X_1X_2}(y_1|x_1,x_2)P_{X_1|C_1}(x_1|c_1),
\end{equation*}
and similarly define Alice's effective channel $W_2(y_2|x_1,c_2)$.

Since each legitimate user knows its own transmitted signal, he/she will treat it as side information while decoding the other user's message. Let $U_j^{1:N}=C_j^{1:N}\mathbf{G}_N$ and $V_j^{1:N}=X_j^{1:N}\mathbf{G}_N$ for $j=1,2$. For $\delta_N=2^{-N^\beta}$ with $\beta \in (0,1/2)$, define
\begin{align}
\mathcal{H}^{(N)}_{C_1|X_2}&\triangleq\{i\in [N]:Z(U_1^i|X_2^{1:N},U_1^{1:i-1})\geq 1-\delta_N\},\nonumber\\
\mathcal{L}^{(N)}_{C_1|Y_1X_2}&\triangleq\{i\in [N]:Z(U_1^i|Y_1^{1:N},X_2^{1:N},U_1^{1:i-1})\leq \delta_N\}, \label{G1}\\
\mathcal{H}^{(N)}_{C_2|X_1}&\triangleq\{i\in [N]:Z(U_2^i|X_1^{1:N},U_2^{1:i-1})\geq 1-\delta_N\},\nonumber\\
\mathcal{L}^{(N)}_{C_2|Y_2X_1}&\triangleq\{i\in [N]:Z(U_2^i|Y_2^{1:N},X_1^{1:N},U_2^{1:i-1})\leq \delta_N\}. \label{G2}
\end{align}
Since two users' messages are independent, we have
\begin{align*}
\mathcal{H}^{(N)}_{C_1|X_2}&=\mathcal{H}^{(N)}_{C_1}\triangleq\{i\in [N]:Z(U_1^i|U_1^{1:i-1})\geq 1-\delta_N\},\\
\mathcal{H}^{(N)}_{C_2|X_1}&=\mathcal{H}^{(N)}_{C_2}\triangleq\{i\in [N]:Z(U_2^i|U_2^{1:i-1})\geq 1-\delta_N\}.
\end{align*}

For conventional two-way communication without secrecy requirement, the information bit sets are defined as
\begin{equation*}
	\mathcal{G}_1=\mathcal{H}^{(N)}_{C_1|X_2}\cap \mathcal{L}^{(N)}_{C_1|Y_1X_2},~~
	\mathcal{G}_2=\mathcal{H}^{(N)}_{C_2|X_1}\cap \mathcal{L}^{(N)}_{C_2|Y_2X_1}.
\end{equation*}
Since such two polar codes can be seen as MAC polar codes designed for corner points, from (\ref{MAC:rate}) we have
\begin{equation}
\label{reliablerate}
\lim\limits_{N\rightarrow \infty}\frac{1}{N}|\mathcal{G}_1|=I(Y_1;C_1|X_2),~~
\lim\limits_{N\rightarrow \infty}\frac{1}{N}|\mathcal{G}_2|=I(Y_2;C_2|X_1).
\end{equation}

To generate the final codeword, one can transmit $C_j^{1:N}$ ($j=1,2$) through a virtual channel with transition probability $P_{X_j|C_j}(x_j|c_j)$, known as channel prefixing. In practice, channel prefixing may not be done perfectly. In this paper, we consider $X_j$ and $C_j$ ($j=1,2$) as two correlated sources and use polar coding to do channel prefixing. Define
\begin{equation}
\mathcal{H}^{(N)}_{X_j|C_j}\triangleq\{i\in [N]:Z(V_j^i|C_j^{1:N},V_j^{1:i-1})\geq 1-\delta_N\}.
\end{equation}
Once $C_j^{1:N}$ is determined, $X_j^{1:N}$ can be obtained as follows:
\begin{itemize}
	\item $\{v_j^i\}_{i\in \mathcal{H}^{(N)}_{X_j|C_j}}$ are filled with uniformly distributed random bits,
	\item $\{v_j^i\}_{i\in (\mathcal{H}^{(N)}_{X_j|C_j})^C}$ are assigned by random mappings:
	\begin{equation*}
	v_j^i=\begin{cases}
	0 ~~\text{  w.p. } P_{V_j^i|C_j^{1:N}V_j^{1:i-1}}(0|c_j^{1:N},v_j^{1:i-1}),\\
	1 ~~\text{  w.p. } P_{V_j^i|C_j^{1:N}V_j^{1:i-1}}(1|c_j^{1:N},v_j^{1:i-1}).
	\end{cases}
	\end{equation*} 
	\item Compute $x_j^{1:N}=v_j^{1:N}\mathbf{G}_N$. 
\end{itemize}

\subsubsection{Polarization of the Eavesdropper Channel}
Eve's effective channel is defined as
\begin{equation*}
W_e(y_e|c_1,c_2)\triangleq \sum_{x_1}\sum_{x_2}P_{Y_e|X_1X_2}(y_e|x_1,x_2)P_{X_1|C_1}(x_1|c_1)P_{X_2|C_2}(x_2|c_2),
\end{equation*}
the achievable rate region of which is given by
\begin{equation}
\label{eveeffect}
\mathcal{R}(W_e) = \left\lbrace
\begin{matrix}
\left(
\begin{array}{ccc}
R_{1}\\
R_{2}
\end{array}
\right)&
\left|
\begin{array}{ccc}
\begin{array}{ccc}
0 \leq R_1 \leq I(C_1;Y_e)\\
0 \leq R_2 \leq I(C_2;Y_e)\\
R_1+R_2 \leq I(C_1,C_2;Y_e)
\end{array}
\end{array}\right.
\end{matrix}
\right\rbrace.
\end{equation}

For an arbitrary point $\mathbf{P}_S=(R_{S1},R_{S2})$ on the dominant face of $\mathcal{R}_S(P_{Y_1Y_2Y_e|X_1X_2})$, let
\begin{equation*}
R_{e1}=I(Y_1;C_1|X_2)-R_{S1},~~
R_{e2}=I(Y_2;C_2|X_1)-R_{S2}.
\end{equation*}
Obviously $\mathbf{P}_E\triangleq(R_{e1},R_{e2})$ is on the dominant face of $\mathcal{R}(W_e)$. Let $S^{2N}$ be the permutation of $U_1^{1:N}U_2^{1:N}$ that achieves $\mathbf{P}_E$ in $W_e$. For $j=1,2$, define $\mathcal{S}_{U_j}\triangleq \{i\in[2N]:S^i\in U_j^{1:N}\}$ and let $f_j(i):[N]\rightarrow \mathcal{S}_{U_j}$ be the mapping from indices of $U_j^{1:N}$ to those of $S^{\mathcal{S}_{U_j}}$. Then we can define the following polarized sets from Eve's point of view:
\begin{align}
\mathcal{H}^{(N)}_{S_{U_j}}&\triangleq \big{\{}i\in [N]:Z(S^{f_j(i)}|S^{1:f_j(i)-1})\geq 1-\delta_N \big{\}},\nonumber\\
\mathcal{H}^{(N)}_{S_{U_j}|Y_e}&\triangleq \big{\{}i\in [N]:Z(S^{f_j(i)}|Y_e^{1:N}, S^{1:f_j(i)-1})\geq 1-\delta_N \big{\}},\label{Be}\\
\mathcal{L}^{(N)}_{S_{U_j}|Y_e}&\triangleq \big{\{}i\in [N]:Z(S^{f_j(i)}|Y_e^{1:N}, S^{1:f_j(i)-1})\leq \delta_N \big{\}}. \label{Ge}
\end{align}
Since two users' messages are independent from each other, we have
\begin{equation*}
\mathcal{H}^{(N)}_{S_{U_j}}=\mathcal{H}^{(N)}_{C_j}.
\end{equation*}

Note that we do not assume how Eve decodes by using a specific permutation $S^{2N}$. The choice of $S^{2N}$ only determines the secrecy rate allocation between two users. We will show in the next subsection that our scheme satisfies the strong secrecy criterion whichever permutation we use.

\subsubsection{Polar Coding for the Two-Way Wiretap Channel}
\label{S:PCTW}
\begin{figure}[tb]
	\centering
	\includegraphics[width=8cm]{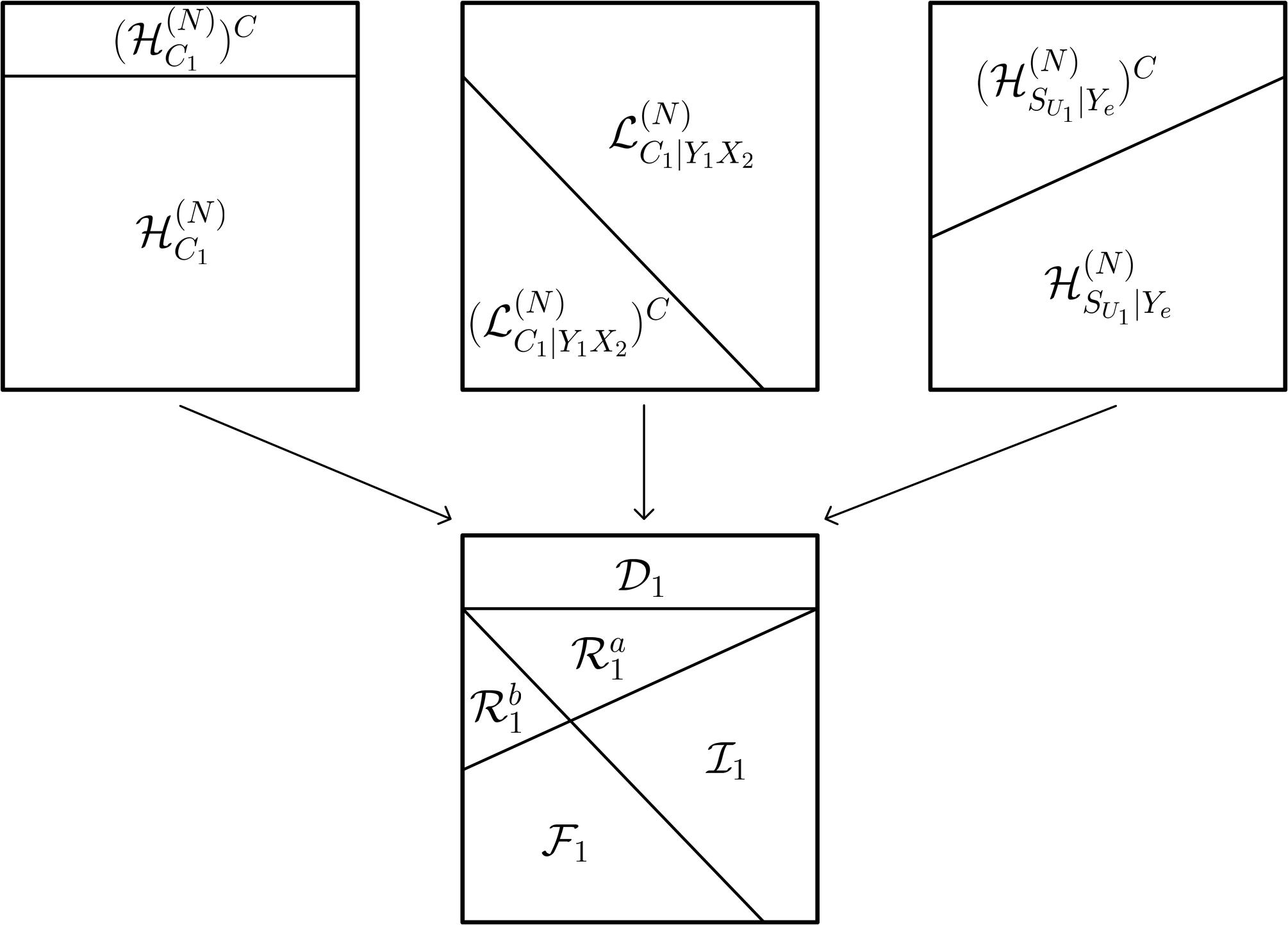}
	\caption{Code construction for Alice.} \label{fig:codecons}
\end{figure}
Define the following index sets for Alice,
\begin{equation}
\begin{aligned}
\label{PART1}
\mathcal{I}_1&=\mathcal{H}^{(N)}_{C_1} \cap \mathcal{L}^{(N)}_{C_1|Y_1X_2} \cap \mathcal{H}^{(N)}_{S_{U_1}|Y_e},\\
\mathcal{F}_1&=\mathcal{H}^{(N)}_{C_1} \cap \big{(}\mathcal{L}^{(N)}_{C_1|Y_1X_2}\big{)}^C \cap \mathcal{H}^{(N)}_{S_{U_1}|Y_e},\\
\mathcal{R}_1^a&=\mathcal{H}^{(N)}_{C_1} \cap \mathcal{L}^{(N)}_{C_1|Y_1X_2} \cap \big{(}\mathcal{H}^{(N)}_{S_{U_1}|Y_e}\big{)}^C,\\
\mathcal{R}_1^b&=\mathcal{H}^{(N)}_{C_1} \cap \big{(}\mathcal{L}^{(N)}_{C_1|Y_1X_2}\big{)}^C \cap \big{(}\mathcal{H}^{(N)}_{S_{U_1}|Y_e}\big{)}^C,\\
\mathcal{D}_1&=\big{(}\mathcal{H}^{(N)}_{C_1}\big{)}^C,
\end{aligned}
\end{equation}
as illustrated in Fig. \ref{fig:codecons}, and similarly define $\mathcal{I}_2$, $\mathcal{F}_2$, $\mathcal{R}_2^a$, $\mathcal{R}_2^b$ and $\mathcal{D}_2$ for Bob. We take $U_1^{1:N}$ as an example to show the code design. Since $\{u_1^i\}_{i\in \mathcal{I}_1}$ are reliable to Bob, but very unreliable to Eve, they can carry secret information. $\{u_1^i\}_{i\in \mathcal{F}_1}$ are unreliable to both of them, thus should be filled with frozen bits. $\{u_1^i\}_{i\in \mathcal{R}_1^a}$ are reliable to both of them, therefore should not carry any secret information. Instead, they will be filled with uniformly distributed random bits. $\{u_1^i\}_{i\in \mathcal{R}_1^b}$ are reliable to Eve but unreliable to the Bob, which poses a problem to the code design. They should serve as frozen bits for Bob while being secured from Eve. A commonly adopted method to solve this problem is the chaining method \cite{hassani2014universal}, which will be described in detail below. $\{u^i\}_{i\in \mathcal{D}_1}$ are the almost deterministic bits to be generated by random mappings.

The key point of the assignment for bits in $\mathcal{R}_1^b$ and $\mathcal{R}_2^b$ is to find a way to ensure their randomness with respect to the eavesdropper and definiteness with respect to the legitimate users simultaneously. Suppose $|\mathcal{I}_1|>|\mathcal{R}_1^b|$ and $|\mathcal{I}_2|>|\mathcal{R}_2^b|$ (corresponding to the positive secrecy rate case). Choose a subset $\mathcal{I}_1^b$ of $\mathcal{I}_1$ such that $|\mathcal{I}_1^b|=|\mathcal{R}_1^b|$, and a subset $\mathcal{I}_2^b$ of $\mathcal{I}_2$ such that $|\mathcal{I}_2^b|=|\mathcal{R}_2^b|$. Denote $\mathcal{I}_1^a=\mathcal{I}_1\setminus \mathcal{I}_1^b$ and $\mathcal{I}_2^a=\mathcal{I}_2\setminus \mathcal{I}_2^b$. Consider a series of $m$ transmission blocks. $\mathcal{I}_1^b$ (resp. $\mathcal{I}_2^b$) in the $k$th ($1\leq k<m$) block is chained to $\mathcal{R}_1^b$ (resp. $\mathcal{R}_2^b$) in the $(k+1)$th block in the sense that bits in them share the same value. In the $k$th block, Alice and Bob can decode $\mathcal{I}_1^b$ and $\mathcal{I}_2^b$ respectively while Eve can not, which will provide bit values for $\mathcal{R}_1^b$ and $\mathcal{R}_2^b$ in the $(k+1)$th block and serve as frozen bits. To initiate the transmission, Alice and Bob should share two secret seed sequences of length $|\mathcal{R}_1^b|$ and $|\mathcal{R}_2^b|$ respectively. The chaining scheme is illustrated in Fig. \ref{fig:simul}. The seed rate of this scheme, $$R_{seed}=\frac{|\mathcal{R}_1^b|+|\mathcal{R}_2^b|}{2mN},$$ can be made arbitrarily small by choosing sufficiently large $m$. Details of the encoding and decoding procedures are as follows.
\begin{figure}[tb]
	\centering
	\includegraphics[width=11cm]{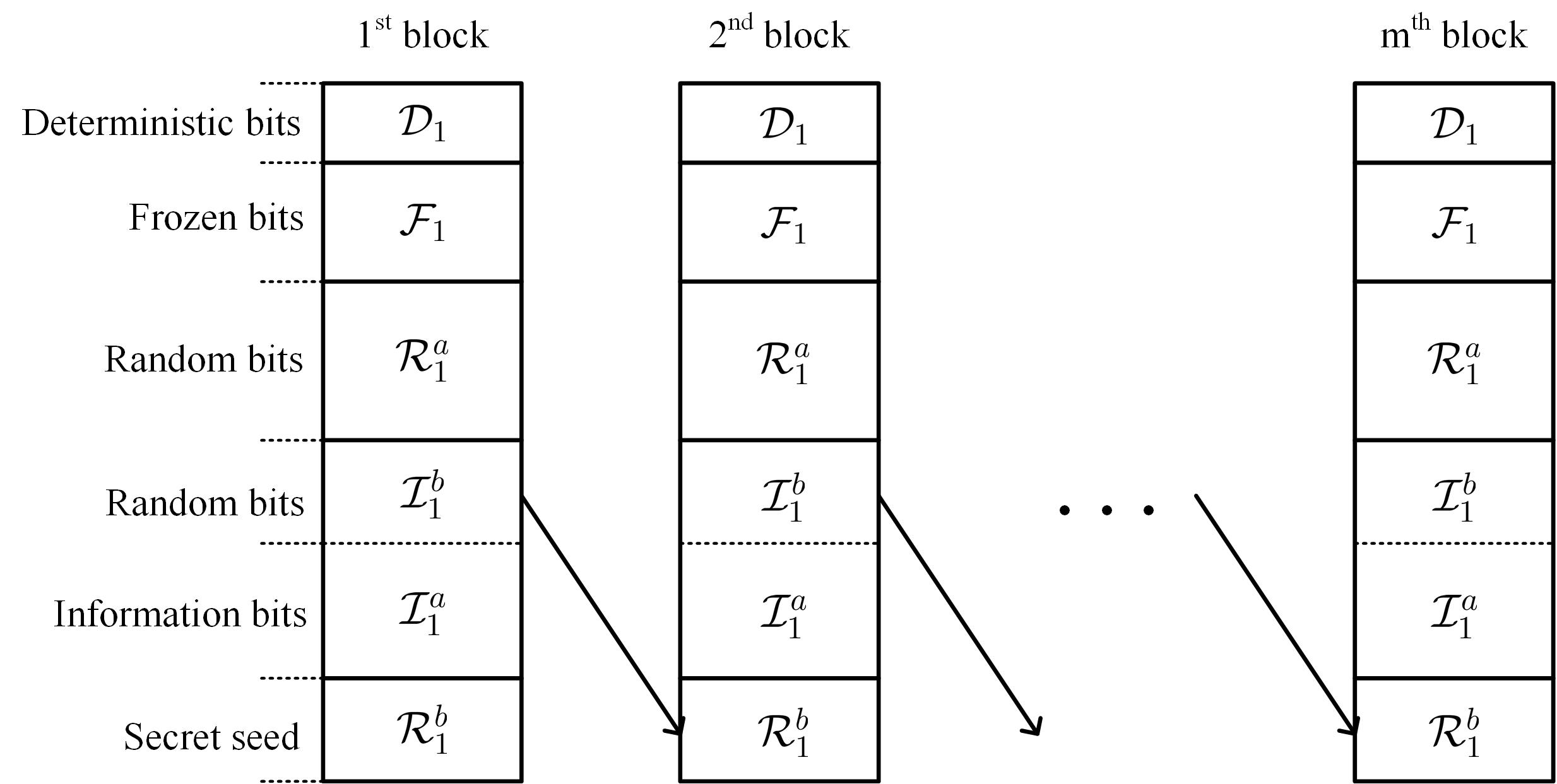}
	\caption{The chaining scheme.} \label{fig:simul}
\end{figure}

\textbf{Encoding:}

In the 1st block, Alice encodes her message as follows:
\begin{itemize}
	\item $\{u_1^i\}_{i\in \mathcal{I}_1^a}$ carry uniformly distributed secret information bits,
	\item $\{u_1^i\}_{i\in \mathcal{R}_1^a\cup \mathcal{I}_1^b}$ are filled with uniformly distributed random bits,
	\item $\{u_1^i\}_{i\in \mathcal{R}_1^b}$ carry uniformly distributed secret seed shared only between two users,
	\item $\{u_1^i\}_{i\in \mathcal{F}_1}$ are filled with uniformly distributed frozen bits (known by everyone, including Eve),
	\item $\{u_1^i\}_{i\in \mathcal{D}_1}$ are assigned by random mappings:
	\begin{equation*}
	u_1^i=\begin{cases}
	0 ~~\text{  w.p. } P_{U^i|U_1^{1:i-1}}(0|u_1^{1:i-1}),\\
	1 ~~\text{  w.p. } P_{U^i|U_1^{1:i-1}}(1|u_1^{1:i-1}).
	\end{cases}
	\end{equation*}
	\item The final codeword is generated as described in Section \ref{SectionIV1}. 
	\item $\{u_1^i\}_{i\in (\mathcal{H}_{C_1}^{(N)})^C\cap (\mathcal{L}_{C_1|Y_1X_2}^{(N)})^C}$ is separately and secretly transmitted to Bob with some reliable error-correcting code.
\end{itemize}

In the $k$th $(1<k\leq m)$ block, $\{u_1^i\}_{i\in \mathcal{R}_1^b}$ are assigned with the same value as $\{u_1^i\}_{i\in \mathcal{I}_1^b}$ in the $(k-1)$th block, and the rest bits are encoded in the same way as in the 1st block.

Bob encodes its message similarly by replacing subscript 1 by 2.

Although our scheme requires separate secret communications between two users, we will show in the next subsection that the rate of them vanishes as the blocklength goes large.

\textbf{Decoding:}

Having received $y_1^{1:N}$ and recovered $\{u_1^i\}_{i\in (\mathcal{H}_{C_1}^{(N)})^C\cap (\mathcal{L}_{C_1|Y_1X_2}^{(N)})^C}$, Bob decodes Alice's message as follows:

In the 1st block, 
\begin{equation*}
\bar{u}_1^{i}=
\begin{cases}
u_1^i,&\text{if } i\in (\mathcal{L}_{C_1|Y_1X_2}^{(N)})^C\\
\arg\max_{u\in\{0,1\}}P_{U_1^{i}|Y_1^{1:N}X_2^{1:N}U_1^{1:i-1}}(u|y_1^{1:N},x_2^{1:N},\bar{u}_1^{1:i-1}),&\text{if } i\in \mathcal{L}_{C_1|Y_1X_2}^{(N)}
\end{cases}.
\end{equation*}

In the $k$th $(1<k\leq m)$ block, $\{\bar{u}_1^i\}_{i\in \mathcal{R}_1^b}$ are deduced from $\{\bar{u}_1^i\}_{i\in \mathcal{I}_1^b}$ in the $(k-1)$th block, and the rest bits are decoded same as in the 1st block.

Alice decodes Bob's message similarly by swapping subscripts 1 and 2.

\subsection{Performance}

\subsubsection{Total Variation Distance}
First, we show that the induced joint distribution by our encoding scheme is asymptotically indistinguishable from the target one. The target joint distribution of polarized variables $U_1^{1:N}V_1^{1:N}U_2^{1:N}V_2^{1:N}$ can be decomposed as
\begin{equation*}
\begin{aligned}
&~~~~P_{U_1^{1:N}V_1^{1:N}U_2^{1:N}V_2^{1:N}}(u_1^{1:N},v_1^{1:N},u_2^{1:N},v_2^{1:N})\\
&=P_{U_1^{1:N}}(u_1^{1:N})P_{V_1^{1:N}|U_1^{1:N}}(v_1^{1:N}|u_1^{1:N})P_{U_2^{1:N}}(u_2^{1:N})P_{V_2^{1:N}|U_2^{1:N}}(v_2^{1:N}|u_2^{1:N})\\
&=\prod_{i=1}^{N}P(u_1^i|u_1^{1:i-1})P(v_1^i|v_1^{1:i-1},u_1^{1:N})P(u_2^i|u_2^{1:i-1})P(v_2^i|v_2^{1:i-1},u_2^{1:N}).
\end{aligned}
\end{equation*}
According to our encoding rules, the induced joint distribution is
\begin{equation*}
\begin{aligned}
&~~~~Q_{U_1^{1:N}V_1^{1:N}U_2^{1:N}V_2^{1:N}}(u_1^{1:N},v_1^{1:N},u_2^{1:N},v_2^{1:N})\\
&=\prod_{i=1}^{N}Q(u_1^i|u_1^{1:i-1})Q(v_1^i|v_1^{1:i-1},u_1^{1:N})Q(u_2^i|u_2^{1:i-1})Q(v_2^i|v_2^{1:i-1},u_2^{1:N}),
\end{aligned}
\end{equation*}
where for $j=1,2$,
\begin{equation*}
Q(u_j^i|u_j^{1:i-1})\triangleq 
\begin{cases}
\frac{1}{2}, &\text{if}~~i\in \mathcal{H}^{(N)}_{C_j},\\
P(u_j^i|u_j^{1:i-1}), &\text{otherwise}.
\end{cases}
\end{equation*}
from our encoding scheme, and
\begin{equation*}
Q(v_j^i|v_j^{1:i-1},u_j^{1:N})\triangleq 
\begin{cases}
\frac{1}{2}, &\text{if}~~i\in \mathcal{H}^{(N)}_{X_j|C_j},\\
P(v_j^i|v_j^{1:i-1},u_j^{1:N}), &\text{otherwise}.
\end{cases}
\end{equation*}
from our channel prefixing scheme. From \cite[Lemma 5]{goela2015broadcast}\footnote{Although there are differences in the number of random variables and encoding rules between the scheme in \cite[Lemma 5]{goela2015broadcast} and ours, one can readily verify from their proof that this conclusion still holds if we replace $\mathcal{M}_1^{(n)}$ and $\mathcal{M}_2^{(n)}$ in \cite[Lemma 5]{goela2015broadcast} respectively with $\mathcal{H}^{(N)}_{X_j|C_j}$ and $\mathcal{H}^{(N)}_{C_j}$ defined in this paper and apply the chain rule on the Kullback-Leibler divergence.} we have
\begin{equation}
\parallel P_{U_1^{1:N}V_1^{1:N}U_2^{1:N}V_2^{1:N}}-Q_{U_1^{1:N}V_1^{1:N}U_2^{1:N}V_2^{1:N}}\parallel \leq O(N2^{-N^\beta}),
\end{equation}
where $\parallel P-Q \parallel$ denotes the total variation distance between distributions $P$ and $Q$. Since $U_j^{1:N}$ and $V_j^{1:N}$ are linear transformations of $C_j^{1:N}$ and $X_j^{1:N}$, respectively, we can readily derive from the definition of total variational distance that
\begin{align}
\parallel P_{all}-Q_{all}\parallel \leq O(N2^{-N^\beta}),\label{TVD}
\end{align}
where $P_{all}$ is the target joint distribution of random variables $C_1^{1:N}X_1^{1:N}C_2^{1:N}X_2^{1:N}Y_1^{1:N}Y_2^{1:N}Y_e^{1:N}$, and $Q_{all}$ is the induced one by our encoding scheme.

\subsubsection{Reliability}
Let $P_{e1}^{(k)}$ and $P_{e2}^{(k)}$ respectively be the block error probability of Bob's and Alice's decoder in the $k$th block under the assumption that the exact value of all frozen bits (including those to be deduced from the previous block) is provided.  From \cite{honda2013asymmetric} we have
\begin{align*}
P_{e1}^{(k)}&\leq \parallel P_{C_1^{1:N}X_1^{1:N}C_2^{1:N}X_2^{1:N}Y_1^{1:N}}-Q_{C_1^{1:N}X_1^{1:N}C_2^{1:N}X_2^{1:N}Y_1^{1:N}}\parallel+\sum_{i \in \mathcal{L}_{C_1|Y_1X_2}^{(N)}} Z(U_1^i|Y_1^{1:N},X_2^{1:N},U_1^{1:i-1})\\
&= O(N2^{-N^\beta}),
\end{align*}
and similarly $P_{e2}^{(k)}\leq  O(N2^{-N^\beta})$. Then the error probability of the overall $m$ transmission blocks can be upper bounded by
\begin{equation}
\label{ber}
P_e\leq \sum_{k=1}^m P_{e1}^{(k)}+\sum_{k=1}^mP_{e2}^{(k)}= O(mN2^{-N^\beta}).
\end{equation}

\subsubsection{Secrecy}
\label{Section:secrecy}
In the $k$th $(1\leq k \leq m)$ block, denote Alice's secret message bits at $\mathcal{I}_1^a$ by $\mathbf{M}_{1,k}$, random bits at $\mathcal{I}_1^b$ by $\mathbf{E}_{1,k}$, and frozen bits at $\mathcal{F}_1$ by $\mathbf{F}_{1,k}$. Bits at $\mathcal{R}_1^b$ in the $k$th block are equal to those at $\mathbf{E}_{1,k-1}$, with $\mathbf{E}_{1,0}$ being the secret seed. $\mathbf{M}_{2,k}$, $\mathbf{E}_{2,k}$ and $\mathbf{F}_{2,k}$ are similarly defined for Bob's message. Eve's channel output in the $k$th block is denoted by $\mathbf{Y}_{e,k}$. For brevity, denote $\mathbf{M}_k\triangleq (\mathbf{M}_{1,k},\mathbf{M}_{2,k})$, $\mathbf{E}_k\triangleq (\mathbf{E}_{1,k},\mathbf{E}_{2,k})$, $\mathbf{F}_k\triangleq \{\mathbf{F}_{1,k},\mathbf{F}_{2,k}\}$, $\mathbf{M}^k\triangleq \{\mathbf{M}_1,...,\mathbf{M}_k\}$,
$\mathbf{E}^k\triangleq \{\mathbf{E}_1,...,\mathbf{E}_k\}$, $\mathbf{F}^k\triangleq \{\mathbf{F}_1,...,\mathbf{F}_k\}$, and $\mathbf{Y}_e^k\triangleq \{\mathbf{Y}_{e,1},...,\mathbf{Y}_{e,k}\}$.

\newtheorem{lemma}{Lemma}
\begin{lemma}
	\label{lemma:1}
	For any $k\in [m]$, we have
	\begin{equation}
	I(\mathbf{M}_k,\mathbf{E}_k;\mathbf{Y}_{e,k}|\mathbf{F}_k)=O(N^3 2^{-N^{\beta}}).
	\end{equation}
\end{lemma}

\begin{proof}
	Let $t=|\mathcal{I}_1|+|\mathcal{I}_2|$ and $w=|\mathcal{F}_1|+|\mathcal{F}_2|$. Denote $\{a_1,a_2,...,a_t\}=\{f_1(i_1),f_2(i_2): i_1 \in \mathcal{I}_1,i_2 \in \mathcal{I}_2\}$ with $a_1<...<a_t$, $\{b_1,b_2,...,b_w\}=\{f_1(i_1),f_2(i_2): i_1 \in \mathcal{F}_1,i_2 \in \mathcal{F}_2\}$ with $b_1<...<b_w$, and $\{c_1,c_2,...,c_{t+w}\}=\{a_1,...,a_t,b_1,...,b_w\}$ with $c_1<...<c_{t+w}$. Then we have
	\begin{align}
	&~~~~I(\mathbf{M}_k,\mathbf{E}_k;\mathbf{Y}_{e,k}|\mathbf{F}_k)\nonumber\\
	&=H(\mathbf{M}_k,\mathbf{E}_k|\mathbf{F}_k)-H(\mathbf{M}_k,\mathbf{E}_k|\mathbf{Y}_{e,k},\mathbf{F}_k)\nonumber\\
	&=H(\mathbf{M}_k,\mathbf{E}_k)-H(\mathbf{M}_k,\mathbf{E}_k,\mathbf{F}_k|\mathbf{Y}_{e,k})+H(\mathbf{F}_k|\mathbf{Y}_{e,k})\nonumber\\
    &=\sum_{i=1}^t H(S^{a_i}|S^{a_1},...,S^{a_{i-1}})-\sum_{i=1}^{t+w} H(S^{c_i}|\mathbf{Y}_{e,k},S^{c_1},...,S^{c_{i-1}})+\sum_{i=1}^w H(S^{b_i}|\mathbf{Y}_{e,k},S^{b_1},...,S^{b_{i-1}})\nonumber\\
	&\leq \sum_{i=1}^{t+w}\big{(}1-H(S^{c_i}|\mathbf{Y}_{e,k},S^{1:c_{i}-1})\big{)}\label{lemma2-1},
	\end{align}
	where (\ref{lemma2-1}) holds because $H(S^{c_i}|\mathbf{Y}_{e,k},S^{c_1},...,S^{c_{i-1}})\geq H(S^{c_i}|\mathbf{Y}_{e,k},S^{1:c_{i}-1})$. Note that the entropy here is calculated under the induced distribution by our encoding scheme. To estimate $I(\mathbf{M}_k,\mathbf{E}_k;\mathbf{Y}_{e,k}|\mathbf{F}_k)$ correctly, let $H_P(S^{c_i}|\mathbf{Y}_{e,k},S^{1:c_{i}-1})$ denote the entropy under the target distribution $P_{C_1^{1:N}X_1^{1:N}C_2^{1:N}X_2^{1:N}Y_1^{1:N}Y_2^{1:N}Y_e^{1:N}}$. According to (\ref{Be}), (\ref{PART1}) and \cite[Proposition 2]{arikan2010source} we have
	\begin{align}
	H_P(S^{c_i}|\mathbf{Y}_{e,k},S^{1:c_{i}-1})\geq Z(S^{c_i}|\mathbf{Y}_{e,k},S^{1:c_{i}-1})^2\geq (1-\delta_N)^2=1-O(2^{-N^{\beta}+1}).\label{Lemma1-1}
	\end{align}
	From \cite[Theorem 17.3.3]{cover2012informtaion} we have
	\begin{align}
	&~~~~|H(S^{c_i}|\mathbf{Y}_{e,k},S^{1:c_{i}-1})-H_P(S^{c_i}|\mathbf{Y}_{e,k},S^{1:c_{i}-1})|\nonumber\\
	&\leq -\parallel P_{C_1^{1:N}X_1^{1:N}C_2^{1:N}X_2^{1:N}Y_e^{1:N}}-Q_{C_1^{1:N}X_1^{1:N}C_2^{1:N}X_2^{1:N}Y_e^{1:N}}\parallel\times\nonumber\\
	&~~~~~~\log\frac{\parallel P_{C_1^{1:N}X_1^{1:N}C_2^{1:N}X_2^{1:N}Y_e^{1:N}}-Q_{C_1^{1:N}X_1^{1:N}C_2^{1:N}X_2^{1:N}Y_e^{1:N}}\parallel}{|\mathcal{Y}_e|^N2^{2N}}\nonumber\\
	&=O(N^2 2^{-N^{\beta}})+O(N^{\beta+1}2^{-N^{\beta}}).\label{Lemma1-2}
	\end{align}
	From (\ref{Lemma1-1}) and (\ref{Lemma1-2}) we have
	\begin{equation*}
	H(S^{c_i}|\mathbf{Y}_{e,k},S^{1:c_{i}-1})\geq 1-O(N^2 2^{-N^{\beta}}).
	\end{equation*}
	Thus,
	\begin{equation}
	I(\mathbf{M}_k,\mathbf{E}_k;\mathbf{Y}_{e,k}|\mathbf{F}_k)\leq O(N^3 2^{-N^{\beta}}).
	\end{equation}	
\end{proof}

Suppose Eve has the knowledge of all frozen bits. The information leakage of our scheme is
\begin{align}
L(N)&=I(\mathbf{M}^m;\mathbf{Y}_e^m|\mathbf{F}^m)\nonumber\\
&\leq I(\mathbf{M}^m,\mathbf{E}_m;\mathbf{Y}_e^m|\mathbf{F}^m)\nonumber\\
&=I(\mathbf{M}^m,\mathbf{E}_m;\mathbf{Y}_{e,m}|\mathbf{F}^m)+I(\mathbf{M}^m,\mathbf{E}_m;\mathbf{Y}_e^{m-1}|\mathbf{F}^m,\mathbf{Y}_{e,m})\nonumber\\
&=I(\mathbf{M}_m,\mathbf{E}_m;\mathbf{Y}_{e,m}|\mathbf{F}^m)+I(\mathbf{M}^m,\mathbf{E}_m;\mathbf{Y}_e^{m-1}|\mathbf{F}^m,\mathbf{Y}_{e,m}) \label{leak-1}\\
&\leq I(\mathbf{M}_m,\mathbf{E}_m;\mathbf{Y}_{e,m}|\mathbf{F}^m)+I(\mathbf{M}^m,\mathbf{E}_m,\mathbf{Y}_{e,m};\mathbf{Y}_e^{m-1}|\mathbf{F}^m)\nonumber\\
&\leq I(\mathbf{M}_m,\mathbf{E}_m;\mathbf{Y}_{e,m}|\mathbf{F}^m)+I(\mathbf{M}^m,\mathbf{E}_{m-1},\mathbf{E}_m,\mathbf{Y}_{e,m};\mathbf{Y}_e^{m-1}|\mathbf{F}^m)\nonumber\\
&= I(\mathbf{M}_m,\mathbf{E}_m;\mathbf{Y}_{e,m}|\mathbf{F}^m)+I(\mathbf{M}^{m-1},\mathbf{E}_{m-1};\mathbf{Y}_e^{m-1}|\mathbf{F}^m), \label{leak-2}\\
&= I(\mathbf{M}_m,\mathbf{E}_m;\mathbf{Y}_{e,m}|\mathbf{F}_m)+I(\mathbf{M}^{m-1},\mathbf{E}_{m-1};\mathbf{Y}_e^{m-1}|\mathbf{F}^{m-1}), \label{leak-3}
\end{align}
where (\ref{leak-1}) and (\ref{leak-2}) are due to the fact that $\mathbf{M}^{m-1}\rightarrow (\mathbf{M}_m,\mathbf{E}_m) \rightarrow \mathbf{Y}_{e,m}$ and $(\mathbf{M}_m,\mathbf{E}_m,\mathbf{Y}_{e,m})\rightarrow (\mathbf{M}^{m-1},\mathbf{E}_{m-1}) \rightarrow \mathbf{Y}_e^{m-1}$ form two Markov chains conditioned on $\mathbf{F}^m$. For (\ref{leak-3}), by the chain rule for mutual information, we have
\begin{align*}
I(\mathbf{M}_m,\mathbf{E}_m;\mathbf{Y}_{e,m}|\mathbf{F}^m)&=I(\mathbf{M}_m,\mathbf{E}_m;\mathbf{Y}_{e,m}|\mathbf{F}_m)+I(\mathbf{M}_m,\mathbf{E}_m;\mathbf{F}^{m-1}|\mathbf{F}_m,\mathbf{Y}_{e,m})\\
&~~~~~~~~-I(\mathbf{M}_m,\mathbf{E}_m;\mathbf{F}^{m-1}|\mathbf{F}_m).
\end{align*}
One can readily verify that no matter we reuse the frozen bits in each block or not, the latter two items in the above equation always equal to 0. Thus, $$I(\mathbf{M}_m,\mathbf{E}_m;\mathbf{Y}_{e,m}|\mathbf{F}^m)=I(\mathbf{M}_m,\mathbf{E}_m;\mathbf{Y}_{e,m}|\mathbf{F}_m).$$
Using a similar analysis we can show that $$I(\mathbf{M}^{m-1},\mathbf{E}_{m-1};\mathbf{Y}_e^{m-1}|\mathbf{F}^m)=I(\mathbf{M}^{m-1},\mathbf{E}_{m-1};\mathbf{Y}_e^{m-1}|\mathbf{F}^{m-1}).$$

By induction hypothesis we have
\begin{equation}
\label{strong-2}
I(\mathbf{M}^m;\mathbf{Y}_e^m|\mathbf{F}^m)\leq \sum_{k=1}^m I(\mathbf{M}_k,\mathbf{E}_k;\mathbf{Y}_{e,k}|\mathbf{F}^k)+I(\mathbf{E}_0;\mathbf{Y}_{e,0}),
\end{equation}
where $I(\mathbf{E}_0;\mathbf{Y}_{e,0})$ is Eve's knowledge about the secret seeds, which should be 0 in a secure coding scheme. From Lemma \ref{lemma:1} we have
\begin{equation}
I(\mathbf{M}^m;\mathbf{Y}_e^m|\mathbf{F}^m)\leq O(mN^3 2^{-N^{\beta}}),
\end{equation}
which means $\lim_{N\rightarrow \infty}L(N)=0$.

\subsubsection{Achievable Rate Region}
\newtheorem{theorem}{Theorem}
  \begin{theorem}
  	\label{theorem:tw}
  	The coding scheme described in Section \ref{S:IV-A} achieves the whole secrecy rate region of the two-way wiretap channel defined in (\ref{SSR}) under the strong secrecy criterion.
  \end{theorem}

  \begin{proof}
  From the previous subsection we can see that the secrecy rate pair of our scheme is
  \begin{equation*}
  \begin{aligned}
    R_1&=\frac{1}{N}\big{(}|\mathcal{I}_1|-|\mathcal{R}_{1}^b|\big{)}\\
    &=\frac{1}{N}\big{(}|\mathcal{I}_1\cup \mathcal{R}_{1}^a|-|\mathcal{R}_{1}^b\cup \mathcal{R}_{1}^a|\big{)}\\
    &=\frac{1}{N}\Big{(}|\mathcal{H}^{(N)}_{C_1} \cap \mathcal{L}^{(N)}_{C_1|Y_1X_2}|-|\mathcal{H}^{(N)}_{C_1} \cap \big{(}\mathcal{H}^{(N)}_{S_{U_1}|Y_e}\big{)}^C|\Big{)},
  \end{aligned}
  \end{equation*}
  \begin{equation*}
  R_2=\frac{1}{N}\Big{(}|\mathcal{H}^{(N)}_{C_2} \cap \mathcal{L}^{(N)}_{C_2|Y_2X_1}|-|\mathcal{H}^{(N)}_{C_2} \cap \big{(}\mathcal{H}^{(N)}_{S_{U_2}|Y_e}\big{)}^C|\Big{)}.
  \end{equation*}
  From (\ref{reliablerate}) we have
  \begin{equation*}
  \lim_{N\rightarrow \infty}\frac{1}{N}|\mathcal{H}^{(N)}_{C_1} \cap \mathcal{L}^{(N)}_{C_1|Y_1X_2}|=I(Y_1;C_1|X_2),~~
  \lim_{N\rightarrow \infty}\frac{1}{N}|\mathcal{H}^{(N)}_{C_2} \cap \mathcal{L}^{(N)}_{C_2|Y_2X_1}|=I(Y_2;C_2|X_1).
  \end{equation*}
  From (\ref{MAC:rate}) we have
  \begin{equation*}
  \lim\limits_{N\rightarrow \infty}\frac{1}{N}|\mathcal{H}^{(N)}_{S_{U_1}}\cap \mathcal{L}^{(N)}_{S_{U_1}|Y_e}|=R_{e1},~~
  \lim\limits_{N\rightarrow \infty}\frac{1}{N}|\mathcal{H}^{(N)}_{S_{U_2}}\cap \mathcal{L}^{(N)}_{S_{U_2}|Y_e}|=R_{e2}.
  \end{equation*}
  Since $\mathcal{H}^{(N)}_{S_{U_j}}=\mathcal{H}^{(N)}_{C_j}$ and $$\lim\limits_{N\rightarrow \infty}\frac{1}{N}|\big{(} \mathcal{H}^{(N)}_{S_{U_j}|Y_e}\cup \mathcal{L}^{(N)}_{S_{U_j}|Y_e} \big{)}^C|=0$$ for $j=1,2$ by \cite[Lemma 1]{chou2015keygen}, we have
  \begin{equation*}
  \lim\limits_{N\rightarrow \infty}\frac{1}{N}|\mathcal{H}^{(N)}_{C_1} \cap \big{(}\mathcal{H}^{(N)}_{S_{U_1}|Y_e}\big{)}^C|=R_{e1},~~
  \lim\limits_{N\rightarrow \infty}\frac{1}{N}|\mathcal{H}^{(N)}_{C_2} \cap \big{(}\mathcal{H}^{(N)}_{S_{U_2}|Y_e}\big{)}^C|=R_{e2}.
  \end{equation*}
  Thus,
  \begin{equation}
  \lim_{N\rightarrow \infty}R_1=I(Y_1;C_1|X_2)-R_{e1}=R_{S1},~~
  \lim_{N\rightarrow \infty}R_2=I(Y_2;C_2|X_1)-R_{e2}=R_{S2}.
  \end{equation}
  Since $(R_{S1},R_{S2})$ is an arbitrary point on the dominant face of $\mathcal{R}_S(P_{Y_1Y_2Y_e|X_1X_2})$, we can say that the whole secrecy rate region is achievable.
  \end{proof}

  \subsubsection{Rate of Shared Randomness}
  We further discuss the rate of shared randomness required in our scheme. As we have shown, the randomized frozen bits only need to be independently and uniformly distributed and can be known by Eve, they thus can be reused over blocks\footnote{To prove it rigorously, with a similar analysis to the proof of \cite[Lemma 8]{chou2015approx}, we can obtain $\parallel P^m-Q^m \parallel \leq O(mN2^{-N^\beta})$, where $P^m$ and $Q^m$ respectively are the target joint distribution and induced joint distribution of random variables $C_1^{1:N}X_1^{1:N}C_2^{1:N}X_2^{1:N}Y_1^{1:N}Y_2^{1:N}Y_e^{1:N}$ in the overall $m$ blocks when the frozen bits are reused. Thus, the error probability still vanishes as $N$ goes to infinity. Since we assume that Eve knows the frozen bits, reusing them also does not harm security.}. In our scheme, suppose each user uses the same frozen bits over $m$ chained blocks. Then the rate of the shared frozen bits is
  \begin{equation}
  R_{F}=\frac{|\mathcal{F}_1|+|\mathcal{F}_2|}{2mN}.
  \end{equation}
  Since $|\mathcal{F}_1|+|\mathcal{F}_2|=O(N)$, we can see that $R_{F}$ can be made arbitrarily small by choosing sufficiently large $m$. Besides, they can actually be generated without sacrificing any transmission rate. For example, all communicators (including the eavesdropper) can use the same pseudorandom generator (PRG) to produce the same pseudorandom frozen bits by inputing the same seed to the PRG, such as the current time.
  
  The rate of the shared almost deterministic bits which need to be separately and secretly exchanged after each transmission block is
  \begin{equation}
  R_{D}=\frac{1}{2N}\big{(}|(\mathcal{H}_{C_1}^{(N)})^C\cap (\mathcal{L}_{C_1|Y_1X_2}^{(N)})^C|+|(\mathcal{H}_{C_2}^{(N)})^C\cap (\mathcal{L}_{C_2|Y_2X_1}^{(N)})^C|\big{)}.
  \end{equation}
  Similar to the point-to-point channel case introduced in Section \ref{S:IIIA}, we have $\lim\limits_{N\rightarrow \infty}R_{D}=0$. 
  
  Recall from Section \ref{S:PCTW} that the secrecy seed rate $R_{seed}$ can also be made arbitrarily small with large $m$. From the above we can conclude that the overall rate of shared randomness required in our scheme can be made negligible by using sufficiently long blocklength and sufficient number of chained blocks.

\section{Special Case: Achieving Weak Secrecy within a Single Transmission Block}
\label{S:V}

  In the traditional one-way wiretap channel, as has been shown in \cite{mahdavifar2011achieving,koyluoglu2012polar}, if the eavesdropper channel is degraded with respect to the main channel, the reliable bit set of a polar code designed for the eavesdropper channel will be a subset of that for the main channel \cite{korada2009polar}, and the secrecy capacity can be achieved under the weak secrecy criterion within a single transmission block. In the two-way wiretap channel, a similar case also exists.  
  
  \begin{definition}
  	Let $P_1:\mathcal{X}_1 \times \mathcal{X}_2 \rightarrow\mathcal{Y}_1$ and $P_2:\mathcal{X}_1 \times \mathcal{X}_2 \rightarrow\mathcal{Y}_2$ be two discrete memoryless multiple access channels, then we say $P_2$ is degraded with respect to $P_1$ (denoted by $P_1 \succ P_2$) if there exists a third channel $P_3:\mathcal{Y}_1\rightarrow\mathcal{Y}_2$ such that $P_2(y_2|x_1,x_2)=\sum_{y_1\in \mathcal{Y}_1}P_1(y_1|x_1,x_2)P_3(y_2|y_1)$ for all $(x_1,x_2)\in \mathcal{X}_1 \times \mathcal{X}_2$ and $y_2\in \mathcal{Y}_2$.
  \end{definition}
  
  \begin{lemma}
  	If $P_{Y_1|C_1,C_2}\succ P_{Y_e|C_1,C_2}$ and $P_{Y_2|C_1,C_2}\succ P_{Y_e|C_1,C_2}$, then we have
  	\begin{equation}
  	\mathcal{L}^{(N)}_{S_{U_1}|Y_e}\subseteq \mathcal{L}^{(N)}_{C_1|Y_1X_2},~~
  	\mathcal{L}^{(N)}_{S_{U_2}|Y_e}\subseteq \mathcal{L}^{(N)}_{C_2|Y_2X_1},
  	\end{equation}
  	where $\mathcal{L}^{(N)}_{S_{U_1}|Y_e}$ and $\mathcal{L}^{(N)}_{S_{U_2}|Y_e}$ are defined in (\ref{Ge}), and $\mathcal{L}^{(N)}_{C_1|Y_1X_2}$ and $\mathcal{L}^{(N)}_{C_2|Y_2X_1}$ are defined in (\ref{G1}) and (\ref{G2}).
  \end{lemma}

  \begin{proof}  
  Since $S^{1:2N}$ is a permutation of $U_1^{1:N}U_2^{1:N}$, we have 
  \begin{equation*}
  Z(S^{f_1(i)}|Y_e^{1:N}, S^{1:f_1(i)-1})\geq Z(U_1^i|Y_e^{1:N},U_2^{1:N},U_1^{1:i-1}).
  \end{equation*}
  And since $(C_1,C_2)\rightarrow (C_1,X_2)\rightarrow Y_1$ forms a Markov chain, we have
  \begin{equation*}
  Z(U_1^i|Y_1^{1:N},U_2^{1:N},U_1^{1:i-1})\geq Z(U_1^i|Y_1^{1:N},X_2^{1:N},U_1^{1:i-1}).
  \end{equation*}  
  If $P_{Y_1|C_1,C_2}\succ P_{Y_e|C_1,C_2}$, then \cite[Lemma 4.7]{korada2009polar}
  \begin{equation*}
  Z(U_1^i|Y_1^{1:N},U_2^{1:N},U_1^{1:i-1})\leq Z(U_1^i|Y_e^{1:N},U_2^{1:N},U_1^{1:i-1}).
  \end{equation*}
  Thus,
  \begin{equation*}
  Z(U_1^i|Y_1^{1:N},X_2^{1:N},U_1^{1:i-1})\leq Z(S^{f_1(i)}|Y_e^{1:N}, S^{1:f_1(i)-1}).
  \end{equation*}
  From the definitions of the polarized sets we can see that $\mathcal{L}^{(N)}_{S_{U_1}|Y_e}\subseteq \mathcal{L}^{(N)}_{C_1|Y_1X_2}$. Similarly we can show that $\mathcal{L}^{(N)}_{S_{U_2}|Y_e}\subseteq \mathcal{L}^{(N)}_{C_2|Y_2X_1}$.

  \end{proof}

  In this special case, we partition indices of $U_1^{1:N}$ into four sets:
  \begin{equation}
  \begin{aligned}
  \label{PART1s}
  \mathcal{I}_1&=\mathcal{H}^{(N)}_{C_1} \cap \mathcal{L}^{(N)}_{C_1|Y_1X_2} \cap \big{(}\mathcal{L}^{(N)}_{S_{U_1}|Y_e}\big{)}^C,\\
  \mathcal{F}_1&=\mathcal{H}^{(N)}_{C_1} \cap \big{(}\mathcal{L}^{(N)}_{C_1|Y_1X_2}\big{)}^C \cap \big{(}\mathcal{L}^{(N)}_{S_{U_1}|Y_e}\big{)}^C,\\
  \mathcal{R}_1&=\mathcal{H}^{(N)}_{C_1} \cap \mathcal{L}^{(N)}_{C_1|Y_1X_2} \cap \mathcal{L}^{(N)}_{S_{U_1}|Y_e},\\
  \mathcal{D}_1&=\big{(}\mathcal{H}^{(N)}_{C_1}\big{)}^C,
  \end{aligned}
  \end{equation}
  as illustrated in Fig. \ref{fig:codecons-s}. Similarly, indices of $U_2^{1:N}$ are partitioned into $\mathcal{I}_2$, $\mathcal{F}_2$, $\mathcal{R}_2$ and $\mathcal{D}_2$.
  
  The coding scheme is then simple. $\mathcal{I}_1$ and $\mathcal{I}_2$ carry information bits, $\mathcal{R}_1$ and $\mathcal{R}_2$ are filled with random bits, $\mathcal{F}_1$ and $\mathcal{F}_2$ carry frozen bits, and $\mathcal{D}_1$ and $\mathcal{D}_2$ are determined with random mappings. Similar to the strong secrecy scheme, a vanishing fraction of the deterministic bits, $\{u_1^i\}_{i\in (\mathcal{H}_{C_1}^{(N)})^C\cap (\mathcal{L}_{C_1|Y_1X_2}^{(N)})^C}$ and $\{u_2^i\}_{i\in (\mathcal{H}_{C_2}^{(N)})^C\cap (\mathcal{L}_{C_2|Y_2X_1}^{(N)})^C}$, are secretly exchanged between Alice and Bob.

  \begin{figure}[tb]
	\centering
	\includegraphics[width=8cm]{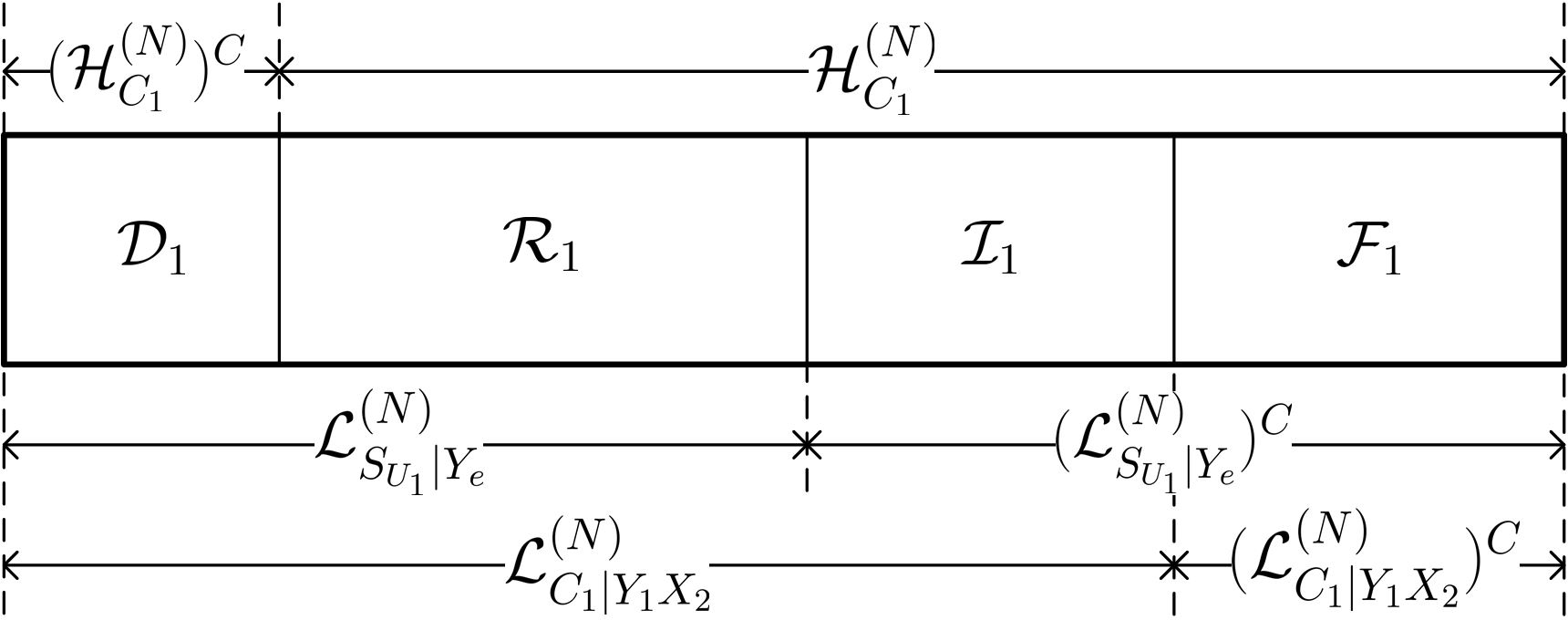}
	\caption{Code construction for Alice in the degraded case.} \label{fig:codecons-s}
  \end{figure}

  \begin{theorem}
  	\label{theorem:tws}
  	If the eavesdropper channel is degraded with respect to both legitimate channels, the coding scheme described in this section achieves all points on the dominant face of the secrecy rate region of the two-way wiretap channel defined in (\ref{SSR}) under the weak secrecy criterion.
  \end{theorem}
  \begin{proof}
  	
  \textit{Reliability:}
  In the weak secrecy scheme, the total variation distance and the error probability can be analyzed similarly to the strong secrecy one (except that there is no chaining). Thus, (\ref{TVD}) also holds, and
  \begin{equation}
  \label{Pe-s}
  \begin{aligned}
  P_e(N)\leq O(N2^{-N^\beta}).
  \end{aligned}
  \end{equation}

  \textit{Secrecy:}
  Since in this section we only consider a single transmission block, we drop the subscripts and superscripts for block numbers in notations used in Section \ref{Section:secrecy}. Similar to (\ref{lemma2-1}), the information leakage can be upper bounded by
  \begin{align}
    I(\mathbf{M};\mathbf{Y}_e|\mathbf{F})&=H(\mathbf{M})-H(\mathbf{M},\mathbf{F}|\mathbf{Y}_e)+H(\mathbf{F}|\mathbf{Y}_e)\nonumber\\
    &\leq \sum_{i\in\mathcal{A}}\big{(}1-H(S^i|\mathbf{Y}_{e,k},S^{1:i-1})\big{)},
  \end{align}
  where $$\mathcal{A}=\{f_1(i_1),f_2(i_2): i_1 \in \mathcal{I}_1\cup \mathcal{F}_1,i_2 \in \mathcal{I}_2\cup \mathcal{F}_2\}.$$ From (\ref{PART1s}) we have $\mathcal{I}_j\cup \mathcal{F}_j=\mathcal{H}^{(N)}_{C_1}\cap \big{(}\mathcal{L}^{(N)}_{S_{U_1}|Y_e}\big{)}^C$ for $j=1,2$. Define 
  $$\mathcal{B}=\{f_1(i_1),f_2(i_2): i_1 \in \mathcal{H}^{(N)}_{C_1}\cap \mathcal{H}^{(N)}_{S_{U_1}|Y_e}, i_2 \in \mathcal{H}^{(N)}_{C_2}\cap \mathcal{H}^{(N)}_{S_{U_2}|Y_e}\},$$
  and let $H_P(S^i|\mathbf{Y}_{e,k},S^{1:i-1})$ denote the entropy of $S^i$ conditioned on $(\mathbf{Y}_{e,k},S^{1:i-1})$ under the target distribution $P_{C_1^{1:N}X_1^{1:N}C_2^{1:N}X_2^{1:N}Y_1^{1:N}Y_2^{1:N}Y_e^{1:N}}$. Then
  \begin{align}
  \sum_{i\in\mathcal{A}}H_P(S^i|\mathbf{Y}_{e,k},S^{1:i-1})&=\sum_{i\in\mathcal{B}}H_P(S^i|\mathbf{Y}_{e,k},S^{1:i-1})+\sum_{i\in\mathcal{A}\setminus\mathcal{B}}H_P(S^i|\mathbf{Y}_{e,k},S^{1:i-1})\nonumber\\
  &\geq\sum_{i\in\mathcal{B}}Z(S^i|\mathbf{Y}_{e,k},S^{1:i-1})^2+\sum_{i\in\mathcal{A}\setminus\mathcal{B}}Z(S^i|\mathbf{Y}_{e,k},S^{1:i-1})^2\label{leak-s1}\\
  &\geq |\mathcal{B}|(1-\delta_N)^2+(|\mathcal{A}|-|\mathcal{B}|)\delta_N^2,\nonumber
  \end{align}
  where (\ref{leak-s1}) holds from \cite[Proposition 2]{arikan2010source}. From (\ref{Lemma1-2}) we have
  \begin{equation*}
  \sum_{i\in\mathcal{A}}H(S^i|\mathbf{Y}_{e,k},S^{1:i-1})\geq |\mathcal{B}|(1-\delta_N)^2+(|\mathcal{A}|-|\mathcal{B}|)\delta_N^2-O(N^3 2^{-N^{\beta}}).
  \end{equation*}
  Thus, the information leakage rate can be upper bounded by
  \begin{align}
  L_R(N)&=\frac{1}{N}I(\mathbf{M};\mathbf{Y}_e|\mathbf{F})\nonumber\\
  &\leq \frac{1}{N}\big{(}|\mathcal{A}|-|\mathcal{B}|+2|\mathcal{B}|\delta_N-|\mathcal{A}|\delta_N^2+O(N^3 2^{-N^{\beta}})\big{)}.
  \end{align}
  Since $|\mathcal{A}|-|\mathcal{B}|=|\mathcal{A}\setminus\mathcal{B}|=o(N)$ by \cite[Lemma 1]{chou2015keygen},  we have
  \begin{equation}
  \label{L-s}
  \lim_{N\rightarrow \infty}L_R(N)=0.
  \end{equation}

  \textit{Achievable rate region:}
  Since $\mathcal{L}^{(N)}_{S_{U_1}|Y_e}\subseteq \mathcal{L}^{(N)}_{C_1|Y_1X_2}$ and $\mathcal{L}^{(N)}_{S_{U_2}|Y_e}\subseteq \mathcal{L}^{(N)}_{C_2|Y_2X_1}$, we have
  \begin{align*}
  R_1=\frac{1}{N}|\mathcal{H}^{(N)}_{C_1} \cap \mathcal{L}^{(N)}_{C_1|Y_1X_2}|-\frac{1}{N}|\mathcal{H}^{(N)}_{C_1} \cap \mathcal{L}^{(N)}_{S_{U_1}|Y_e}|,
  \end{align*}
  \begin{equation*}
  R_2=\frac{1}{N}|\mathcal{H}^{(N)}_{C_2} \cap \mathcal{L}^{(N)}_{C_2|Y_2X_1}|-\frac{1}{N}|\mathcal{H}^{(N)}_{C_2} \cap \mathcal{L}^{(N)}_{S_{U_2}|Y_e}|.
  \end{equation*}
  
  We can then use a similar analysis to the strong secrecy case and show that this scheme achieves all points on the dominant face of $\mathcal{R}_S(P_{Y_1Y_2Y_e|X_1X_2})$. For conciseness we omit it here.
  
  Now we have finished the proof for Theorem \ref{theorem:tws}.

  \end{proof}

\section{Example: Binary Erasure Channels}
\label{S:VI}
In this section, we present an example to show the performance of our scheme. For simplicity, all channels are assumed to be binary erasure MACs, defined as
\begin{equation}
Y=
\begin{cases}
X_1+X_2 & \text{ w.p. } 1-\epsilon\\
? & \text{ w.p. } \epsilon
\end{cases},
\end{equation}
and the channel inputs are assumed to be uniformly distributed. The erasure probabilities of Bob's observed channel $W_1(Y_1|X_1X_2)$, Alice's observed channel $W_2(Y_2|X_1X_2)$ and Eve's observed channel $W_e(Y_1|X_1X_2)$ are $\epsilon_1=0.2$, $\epsilon_2=0.3$ and $\epsilon_e=0.4$ respectively. For the auxiliary random variables in (\ref{SSR}), we consider $C_1=X_1$ and $C_2=X_2$. In this case, the achievable rate region of the eavesdropper MAC is
\begin{equation}
\begin{cases}
&0\leq R_{e1} \leq R_U=1-\epsilon_e=0.6\\
&0\leq R_{e2} \leq R_V=1-\epsilon_e=0.6\\
&R_{e1}+R_{e2} \leq C_{sum}=1.5(1-\epsilon_e)=0.9
\end{cases}.
\end{equation}
Since each user knows its own transmitted message, two legitimate channels can be simplified to two BECs with erasure probabilities $\epsilon_1$ and $\epsilon_2$ respectively. Then the secrecy rate region of this two-way wiretap channel is
\begin{equation*}
	\begin{cases}
		&0\leq R_1 \leq (1-\epsilon_1)-(C_{sum}-R_V)=0.5\\
		&0\leq R_2 \leq (1-\epsilon_2)-(C_{sum}-R_U)=0.4\\
		&R_1+R_2 \leq R_s\triangleq (1-\epsilon_1)+(1-\epsilon_2)-C_{sum}=0.6
	\end{cases}.
\end{equation*}

For BECs, Bhattacharyya parameters $Z(W_{1,N}^{(i)})$ and $Z(W_{2,N}^{(i)})$ can be easily calculated by \cite{arikan2009channel}
\begin{equation*}
Z(W_{j,N}^{(2i-1)})=2Z(W_{j,N/2}^{(i)})-Z(W_{j,N/2}^{(i)})^2,~~
Z(W_{j,N}^{(2i)})=Z(W_{j,N/2}^{(i)})^2,
\end{equation*}
with $Z(W_{j,1}^{(1)})=\epsilon_j$ for $j=1,2$, where $Z(W_{j,N}^{(i)})$ is short for $Z(U_j^i|Y_j^{1:N},U_j^{1:i-1})$.

For the eavesdropper MAC, we take a corner point of its achievable rate region, $(0.6,0.3)$, as an example. The secrecy rate pair in this case is $(0.5,0.1)$.The Bhattacharyya parameters in the corner point case can be easily calculated since the MAC can be split into two single-user channels. For  points between two corner points, a Monte Carlo approach can be used \cite{arikan2009channel}. Let $W_{e1}(Y_e|X_1)$ be the channel from Alice to Eve when $X_2$ is known to Eve, and $W_{e2}(Y_e|X_2)$ the channel from Bob to Eve when $X_1$ is treated as noise. It is easy to verify that Bhattacharyya parameters for $W_{e1}$ and $W_{e2}$ are the same as those for $W_{e1}(Y_e|X_1)$.
\begin{figure}[tb]
	\centering
	\includegraphics[width=10cm]{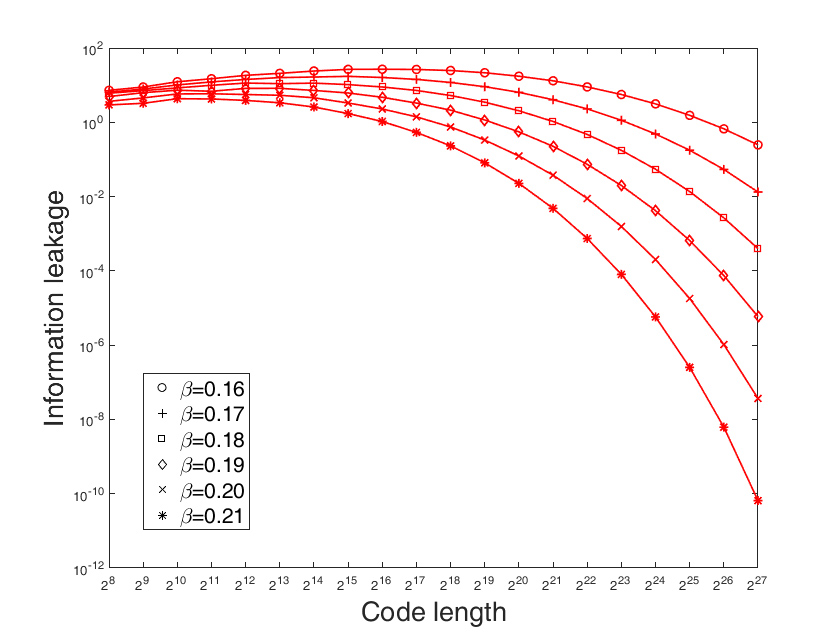}
	\caption{Upper bound for the information leakage.} \label{fig:secrecy}
\end{figure}

\begin{figure}[tb]
	\centering
	\includegraphics[width=10cm]{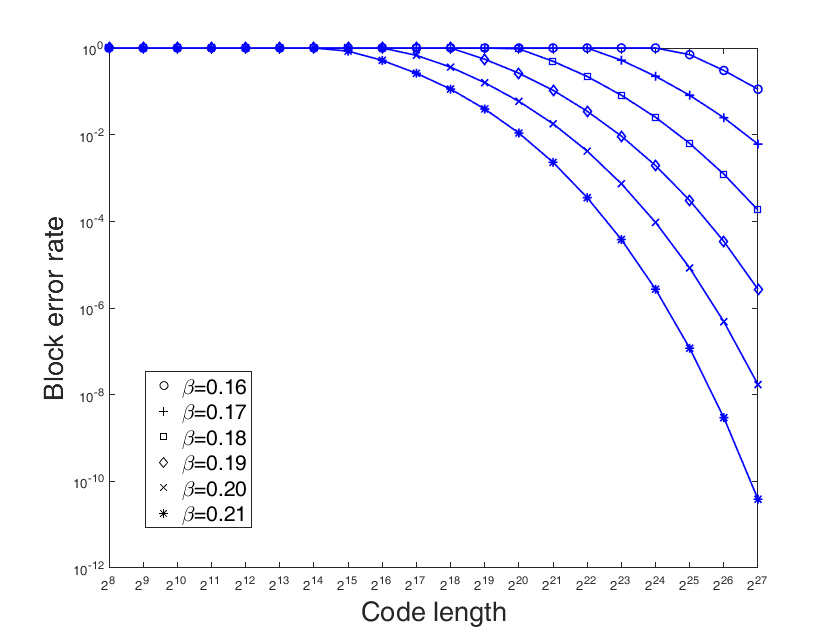}
	\caption{Upper bound for the block error rate.} \label{fig:reliability}
\end{figure}

\begin{figure}[tb]
	\centering
	\includegraphics[width=10cm]{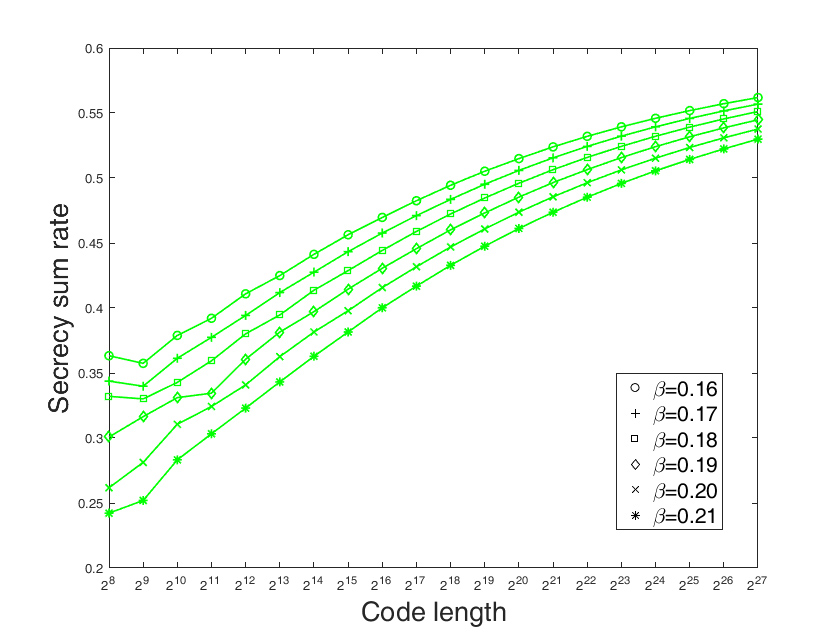}
	\caption{The secrecy sum rate of our proposed scheme.} \label{fig:sumrate}
\end{figure}

The upper bound for information leakage can be deduced from (\ref{lemma2-1}) and (\ref{strong-2}) that
\begin{equation}
\label{strong-z}
I(\mathbf{M}^m;\mathbf{Y}_e^m|\mathbf{F}^m)\leq m\sum_{i\in \mathcal{I}\cup\mathcal{F}}{\big{(}1-Z(S^{i}|Y_e^{1:N}, S^{1:i-1})^2\big{)}},
\end{equation}
where $\mathcal{I}$ and $\mathcal{F}$ respectively are the information bit set and the frozen bit set with respect to $S^{1:2N}$. Once we are able to obtain the Bhattacharyya parameters for all synthesized channels, the upper bound for block error rate can be evaluated by (\ref{ber}), and that for the information leakage can be estimated using (\ref{strong-z}). Since the number of transmission blocks $m$ is only a multiplier when estimating the information leakage and block error rate, we choose $m=1$ without loss of generality.

The parameter $\delta_N=2^{-N^\beta}$ ($0<\beta<0.5$) in the definitions of polarized sets plays an important role in the code design. The larger $\beta$ is, the smaller $L(N)$ and $P_e(N)$ will be for a given $N$. However, the secret sum rate will be smaller correspondingly. In this example, we choose several $\beta$ values and compare their differences in performance. Fig. \ref{fig:secrecy}, \ref{fig:reliability} and \ref{fig:sumrate} respectively show the information leakage upper bound, block error rate upper bound and secrecy sum rate of our proposed scheme. The result meets our theoretical analysis that, as the code length increases, the information leakage and block error rate vanish while the secrecy sum rate approaches $R_s=0.6$.

\section{Conclusion and Discussion}
\label{S:VII}
  In this paper, we have introduced polar codes into the problem of coded cooperative jamming in the two-way wiretap channel, and proposed a strong secrecy-achieving scheme for the general case and a weak secrecy-achieving scheme for the degraded case. The monotone chain rule expansion based MAC polarization method is used to allocate different secrecy rates between two users. How to determine the exact code construction for an arbitrary rate pair in an arbitrary two-way wiretap channel might be a difficult task in our proposed scheme, since it is not clear whether the existing efficient construction methods for point-to-point polar codes can be readily applied to the monotone chain rule expansion based MAC polar codes yet. A remedy for this problem is to use the other type of MAC polarization method mentioned in Section \ref{S:Intro}, for which efficient constructing methods do exist (e.g.,\cite{Tal2012MAC,Pereg2017MAC}). Although there will be some loss in achievable rate region with this approach, the sum capacity is still achievable.

  Although we only considered binary-input case in this paper, the result can be readily extended to arbitrary prime alphabet cases with the result of \cite{sasoglu2011polar}. 

\bibliographystyle{IEEEtran}
\bibliography{Polar_2Way}

\begin{thebibliography}{10}
\providecommand{\url}[1]{#1}
\csname url@samestyle\endcsname
\providecommand{\newblock}{\relax}
\providecommand{\bibinfo}[2]{#2}
\providecommand{\BIBentrySTDinterwordspacing}{\spaceskip=0pt\relax}
\providecommand{\BIBentryALTinterwordstretchfactor}{4}
\providecommand{\BIBentryALTinterwordspacing}{\spaceskip=\fontdimen2\font plus
\BIBentryALTinterwordstretchfactor\fontdimen3\font minus
  \fontdimen4\font\relax}
\providecommand{\BIBforeignlanguage}[2]{{%
\expandafter\ifx\csname l@#1\endcsname\relax
\typeout{** WARNING: IEEEtran.bst: No hyphenation pattern has been}%
\typeout{** loaded for the language `#1'. Using the pattern for}%
\typeout{** the default language instead.}%
\else
\language=\csname l@#1\endcsname
\fi
#2}}
\providecommand{\BIBdecl}{\relax}
\BIBdecl

\bibitem{wyner1975wire}
A.~D. Wyner, ``The wire-tap channel,'' \emph{The Bell System Technical
  Journal}, vol.~54, no.~8, pp. 1355--1387, 1975.

\bibitem{suresh2010strong}
A.~T. Suresh, A.~Subramanian, A.~Thangaraj, M.~Bloch, and S.~W. McLaughlin,
  ``Strong secrecy for erasure wiretap channels,'' in \emph{IEEE Information
  Theory Workshop (ITW)}, 2010, pp. 1--5.

\bibitem{thangaraj2007applications}
A.~Thangaraj, S.~Dihidar, A.~R. Calderbank, S.~W. McLaughlin, and J.~M.
  Merolla, ``Applications of {LDPC} codes to the wiretap channel,'' \emph{IEEE
  Transactions on Information Theory}, vol.~53, no.~8, pp. 2933--2945, 2007.

\bibitem{cheraghchi2012invertible}
M.~Cheraghchi, F.~Didier, and A.~Shokrollahi, ``Invertible extractors and
  wiretap protocols,'' \emph{IEEE Transactions on Information Theory}, vol.~58,
  no.~2, pp. 1254--1274, 2012.

\bibitem{arikan2009channel}
E.~Ar{\i}kan, ``Channel polarization: A method for constructing
  capacity-achieving codes for symmetric binary-input memoryless channels,''
  \emph{IEEE Transactions on Information Theory}, vol.~55, no.~7, pp.
  3051--3073, 2009.

\bibitem{arikan2010source}
------, ``Source polarization,'' in \emph{IEEE International Symposium on
  Information Theory (ISIT)}, 2010, pp. 899--903.

\bibitem{korada2009polar}
S.~B. Korada, ``Polar codes for channel and source coding,'' Ph.D.
  dissertation, {\'E}COLE POLYTECHNIQUE F{\'E}D{\'E}RALE DE LAUSANNE, 2009.

\bibitem{korada2010lossy}
S.~B. Korada and R.~L. Urbanke, ``Polar codes are optimal for lossy source
  coding,'' \emph{IEEE Transactions on Information Theory}, vol.~56, no.~4, pp.
  1751--1768, 2010.

\bibitem{honda2013asymmetric}
J.~Honda and H.~Yamamoto, ``Polar coding without alphabet extension for
  asymmetric models,'' \emph{IEEE Transactions on Information Theory}, vol.~59,
  no.~12, pp. 7829--7838, 2013.

\bibitem{hof2010secrecy}
E.~Hof and S.~Shamai, ``Secrecy-achieving polar-coding,'' in \emph{IEEE
  Information Theory Workshop (ITW)}, 2010, pp. 1--5.

\bibitem{andersson2010nested}
M.~Andersson, V.~Rathi, R.~Thobaben, J.~Kliewer, and M.~Skoglund, ``Nested
  polar codes for wiretap and relay channels,'' \emph{IEEE Communications
  Letters}, vol.~14, no.~8, pp. 752--754, 2010.

\bibitem{mahdavifar2011achieving}
H.~Mahdavifar and A.~Vardy, ``Achieving the secrecy capacity of wiretap
  channels using polar codes,'' \emph{IEEE Transactions on Information Theory},
  vol.~57, no.~10, pp. 6428--6443, 2011.

\bibitem{koyluoglu2012polar}
O.~O. Koyluoglu and H.~El~Gamal, ``Polar coding for secure transmission and key
  agreement,'' \emph{IEEE Transactions on Information Forensics and Security},
  vol.~7, no.~5, pp. 1472--1483, 2012.

\bibitem{sasoglu2013strong}
E.~\c{S}a\c{s}o\u{g}lu and A.~Vardy, ``A new polar coding scheme for strong
  security on wiretap channels,'' in \emph{IEEE International Symposium on
  Information Theory Proceedings (ISIT)}, July 2013, pp. 1117--1121.

\bibitem{wei2016general}
Y.~P. Wei and S.~Ulukus, ``Polar coding for the general wiretap channel with
  extensions to multiuser scenarios,'' \emph{IEEE Journal on Selected Areas in
  Communications}, vol.~34, no.~2, pp. 278--291, 2016.

\bibitem{chou2016broad}
R.~A. Chou and M.~R. Bloch, ``Polar coding for the broadcast channel with
  confidential messages: A random binning analogy,'' \emph{IEEE Transactions on
  Information Theory}, vol.~62, no.~5, pp. 2410--2429, 2016.

\bibitem{gulcu2017wiretap}
T.~C. Gulcu and A.~Barg, ``Achieving secrecy capacity of the wiretap channel
  and broadcast channel with a confidential component,'' \emph{IEEE
  Transactions on Information Theory}, vol.~63, no.~2, pp. 1311--1324, 2017.

\bibitem{Si2016fadingwiretap}
H.~Si, O.~O. Koyluoglu, and S.~Vishwanath, ``Hierarchical polar coding for
  achieving secrecy over state-dependent wiretap channels without any
  instantaneous {CSI},'' \emph{IEEE Transactions on Communications}, vol.~64,
  no.~9, pp. 3609--3623, 2016.

\bibitem{chou2016macwiretap}
R.~A. Chou and A.~Yener, ``Polar coding for the multiple access wiretap channel
  via rate-splitting and cooperative jamming,'' in \emph{IEEE International
  Symposium on Information Theory (ISIT)}, 2016, pp. 983--987.

\bibitem{hajimomeni2016mawc}
M.~Hajimomeni, H.~Aghaeinia, I.~M. Kim, and K.~Kim, ``Cooperative jamming polar
  codes for multiple-access wiretap channels,'' \emph{IET Communications},
  vol.~10, no.~4, pp. 407--415, 2016.

\bibitem{tekin2009mactw}
E.~Tekin and A.~Yener, ``The general {G}aussian multiple-access and two-way
  wiretap channels: Achievable rates and cooperative jamming,'' \emph{IEEE
  Transactions on Information Theory}, vol.~54, no.~6, pp. 2735--2751, June
  2008.

\bibitem{pierrot2011twoway}
A.~Pierrot and M.~Bloch, ``Strongly secure communications over the two-way
  wiretap channel,'' \emph{IEEE Transactions on Information Forensics and
  Security}, vol.~6, no.~3, pp. 595--605, Sept 2011.

\bibitem{Gamal2013secrecy}
A.~El~Gamal, O.~Koyluoglu, M.~Youssef, and H.~El~Gamal, ``Achievable secrecy
  rate regions for the two-way wiretap channel,'' \emph{IEEE Transactions on
  Information Theory}, vol.~59, no.~12, pp. 8099--8114, Dec 2013.

\bibitem{pierrot2012ldpc}
A.~Pierrot and M.~Bloch, ``{LDPC}-based coded cooperative jamming codes,'' in
  \emph{IEEE Information Theory Workshop (ITW)}, Sept 2012, pp. 462--466.

\bibitem{sasoglu2013mac}
E.~\c{S}a\c{s}o\u{g}lu, E.~Telatar, and E.~Yeh, ``Polar codes for the two-user
  multiple-access channel,'' \emph{IEEE Transactions on Information Theory},
  vol.~59, no.~10, pp. 6583--6592, Oct 2013.

\bibitem{abbe2012mmac}
E.~Abbe and I.~Telatar, ``Polar codes for the m-user multiple access channel,''
  \emph{IEEE Transactions on Information Theory}, vol.~58, no.~8, pp.
  5437--5448, Aug 2012.

\bibitem{Nasser2016MAC}
R.~Nasser and E.~Telatar, ``Polar codes for arbitrary {DMC}s and arbitrary
  {MAC}s,'' \emph{IEEE Transactions on Information Theory}, vol.~62, no.~6, pp.
  2917--2936, 2016.

\bibitem{arikan2012sw}
E.~Ar{\i}kan, ``Polar coding for the {S}lepian-{W}olf problem based on monotone
  chain rules,'' in \emph{IEEE International Symposium on Information Theory
  Proceedings (ISIT)}, July 2012, pp. 566--570.

\bibitem{onayscmac}
S.~Onay, ``Successive cancellation decoding of polar codes for the two-user
  binary-input {MAC},'' in \emph{IEEE International Symposium on Information
  Theory Proceedings (ISIT)}, July 2013, pp. 1122--1126.

\bibitem{mahdavifar2016uniform}
H.~Mahdavifar, M.~El-Khamy, J.~Lee, and I.~Kang, ``Achieving the uniform rate
  region of general multiple access channels by polar coding,'' \emph{IEEE
  Transactions on Communications}, vol.~64, no.~2, pp. 467--478, 2016.

\bibitem{hassani2014universal}
S.~H. Hassani and R.~Urbanke, ``Universal polar codes,'' in \emph{2014 IEEE
  International Symposium on Information Theory}, 2014, pp. 1451--1455.

\bibitem{gad2016asymmetric}
E.~E. Gad, Y.~Li, J.~Kliewer, M.~Langberg, A.~A. Jiang, and J.~Bruck,
  ``Asymmetric error correction and flash-memory rewriting using polar codes,''
  \emph{IEEE Transactions on Information Theory}, vol.~62, no.~7, pp.
  4024--4038, 2016.

\bibitem{zheng2016polarIC}
\BIBentryALTinterwordspacing
M.~{Zheng}, C.~{Ling}, W.~{Chen}, and M.~{Tao}, ``{A New Polar Coding Scheme
  for the Interference Channel},'' \emph{ArXiv e-prints}, Aug. 2016. [Online].
  Available: \url{http://arxiv.org/abs/1608.08742}
\BIBentrySTDinterwordspacing

\bibitem{cover2012informtaion}
T.~M. Cover and J.~A. Thomas, \emph{Elements of information theory}.\hskip 1em
  plus 0.5em minus 0.4em\relax John Wiley \& Sons, 2012.

\bibitem{goela2015broadcast}
N.~Goela, E.~Abbe, and M.~Gastpar, ``Polar codes for broadcast channels,''
  \emph{IEEE Transactions on Information Theory}, vol.~61, no.~2, pp. 758--782,
  2015.

\bibitem{chou2015keygen}
R.~A. Chou, M.~R. Bloch, and E.~Abbe, ``Polar coding for secret-key
  generation,'' \emph{IEEE Transactions on Information Theory}, vol.~61,
  no.~11, pp. 6213--6237, 2015.

\bibitem{chou2015approx}
R.~A. Chou, M.~R. Bloch, and J.~Kliewer, ``Polar coding for empirical and
  strong coordination via distribution approximation,'' in \emph{2015 IEEE
  International Symposium on Information Theory (ISIT)}, 2015, pp. 1512--1516.

\bibitem{Tal2012MAC}
I.~Tal, A.~Sharov, and A.~Vardy, ``Constructing polar codes for non-binary
  alphabets and {MAC}s,'' in \emph{2012 IEEE International Symposium on
  Information Theory Proceedings}, 2012, pp. 2132--2136.

\bibitem{Pereg2017MAC}
U.~Pereg and I.~Tal, ``Channel upgradation for non-binary input alphabets and
  {MAC}s,'' \emph{IEEE Transactions on Information Theory}, vol.~63, no.~3, pp.
  1410--1424, 2017.

\bibitem{sasoglu2011polar}
E.~\c{S}a\c{s}o\u{g}lu, ``Polar coding theorems for discrete systems,'' Ph.D.
  dissertation, {\'E}COLE POLYTECHNIQUE F{\'E}D{\'E}RALE DE LAUSANNE, 2011.

\end{thebibliography}

\end{document}